%% file: maths.tex
\documentclass[12pt]{iopart} 

\pdfoutput=1 

\usepackage{hyperref}
\usepackage{graphicx}
\usepackage{color}
\usepackage{etoolbox}

\usepackage{amsthm} 
\newtheorem{lem}{Lemma}
\newtheorem{thm}{Theorem}
\newtheorem{cor}{Corollary}
\usepackage{amsfonts} 

\makeatletter
\def\@mkboth#1#2{}
\newlength\appendixwidth
\preto\appendix{\addtocontents{toc}{\protect\patchl@section}}
\newcommand{\patchl@section}{%
   \settowidth{\appendixwidth}{\textbf{Appendix }}%
   \addtolength{\appendixwidth}{1.5em}%
   \patchcmd{\l@section}{1.5em}{\appendixwidth}{}{\ddt}%
}
\makeatother

\begin{document}

\input{macros}

\title{The variation of the sum of edge lengths in linear arrangements of trees}

\author{Ramon Ferrer-i-Cancho$^1$, Carlos G\'omez-Rodr\'iguez$^2$ and Juan Luis Esteban$^3$}
\address{$^1$ Complexity \& Quantitative Linguistics Lab, LARCA Research Group \\
Departament de Ci\`encies de la Computaci\'o, \\
Universitat Polit\`ecnica de Catalunya, \\
Campus Nord, Edifici Omega, Jordi Girona Salgado 1-3. \\
08034 Barcelona, Catalonia (Spain)}
\address{$^2$ Universidade da Coru\~na, CITIC \\
FASTPARSE Lab, LyS Research Group \\ 
Departamento de Ciencias de la Computaci\'on y Tecnolog\'ias de la Informaci\'on, \\ 
Facultade de Inform\'atica, \\ 
Campus de A Coru\~na, 15071 A Coru\~na, Spain
}
\address{$^3$ 
Departament de Ci\`encies de la Computaci\'o, \\
Universitat Polit\`ecnica de Catalunya, \\
Campus Nord, Edifici Omega, Jordi Girona Salgado 1-3. \\
08034 Barcelona, Catalonia (Spain)}
\ead{rferrericancho@cs.upc.edu, cgomezr@udc.es, esteban@cs.upc.edu}

\begin{abstract}
A fundamental problem in network science is the normalization of the topological or physical distance between vertices, that requires understanding the range of variation of the unnormalized distances. 
Here we investigate the limits of the variation of the physical distance in linear arrangements of the vertices of trees. In particular, we investigate various problems on the sum of edge lengths in trees of a fixed size: the minimum and the maximum value of the sum for specific trees, the minimum and the maximum in classes of trees (bistar trees and caterpillar trees) and finally the minimum and the maximum for any tree. We establish some foundations for research on optimality scores for spatial networks in one dimension.  
\end{abstract}
	



\tableofcontents

\maketitle

\section{Introduction}

\input{introduction}

\section{Research problems and review}

\label{review_section}
\input{review}

\section{The maximum value of $D$}

\input{maximum_sum_of_edge_lengths}

\section{The minimum value of $D$}

\input{minimum_sum_of_edge_lengths}

\section{The maxima of optimality scores}

\input{maxima_of_optimality_scores}

\section{The minimum and the maximum $z$-score}

\input{zeta_score_section}

\section{Discussion}

\input{discussion}

\ack
We are very grateful to L. Alemany-Puig for his careful revision of the manuscript. We also thank two anonymous reviewers for their very valuable feedback. 
RFC is supported by the grant TIN2017-89244-R from MINECO (Ministerio de Econom\'{\i}a, Industria y Competitividad) and the recognition 2017SGR-856 (MACDA) from AGAUR (Generalitat de Catalunya). 
CGR is funded by the European Research Council (ERC), under the European Union's Horizon 2020 research and innovation
programme (FASTPARSE, grant agreement No 714150), the ANSWER-ASAP project (TIN2017-85160-C2-1-R) from MINECO and Xunta de Galicia
(ED431C 2020/11, ED431G2019/01, and an Oportunius program grant to complement ERC grants).
JLE is funded by the grants TIN2016-76573-C2-1-P and PID2019-109137GB-C22 from MINECO.

\appendix

\input{appendix}

\section*{References}

\bibliographystyle{unsrt}

\bibliography{../../../biblio_dt/main,../../../biblio_dt/twoplanaracl,../../../biblio_dt/Ramon,../../../biblio_dt/twoplanaracl_ours,../../../biblio_dt/Ramon_ours,../../../biblio_dt/optimization_in_biology,../../../biblio_dt/bibliography}

\end{document}

%% file: macros.tex
\newcommand{\E}{\mathbb{E}}
\newcommand{\V}{\mathbb{V}} 

\newcommand{\DegreeSecondMoment}{\left< k^2 \right>}
\newcommand{\DegreeSecondMomentLinearTree}{4 - \frac{6}{n}}
\newcommand{\DegreeSecondMomentBalancedBistarTree}{???}
\newcommand{\DegreeSecondMomentQuasistarTree}{n-3 - \frac{6}{n}}
\newcommand{\DegreeSecondMomentStarTree}{n-1}
\newcommand{\LargestDegreeBalancedBistarTree}{\left\lceil \frac{n}{2} \right\rceil} 
\newcommand{\SecondLargestDegreeBalancedBistarTree}{\left\lfloor \frac{n}{2} \right\rfloor} 

\newcommand{\DmaxLinearTreeFloor}{\left\lfloor \frac{n^2}{2} \right\rfloor - 1}
\newcommand{\DmaxLinearTreeMod}{\frac{1}{2}(n^2 - {n \bmod 2}) - 1}
\newcommand{\DmaxBistarTree}{k_1 (n - k_1) + \frac{n}{2}(n - 3) + 1}
\newcommand{\DmaxBalancedBistarTree}{\frac{1}{4}\left(3(n-1)^2 + 1 - n \bmod 2 \right)}
\newcommand{\DmaxQuasistarTree}{\frac{1}{2}(n+3)(n-2)}
\newcommand{\DmaxStarTree}{{n \choose 2}}

\newcommand{\DminLinearTree}{n - 1}

\newcommand{\DminCaterpillarTreeFloorHorton}{n - 1 + \sum_{i=1}^n \left\lfloor \frac{1}{4} (k_i - 1)^2 \right\rfloor}
\newcommand{\DminCaterpillarTreeFloor}{\sum_{i=1}^n \left\lfloor \frac{1}{4} (k_i+1)^2 \right\rfloor - (n - 1)}
\newcommand{\DminBistarTreeFloor}{\left\lfloor \frac{1}{4}(k_1 + 1)^2 \right\rfloor + \left\lfloor \frac{1}{4}(n - k_1 + 1)^2 \right\rfloor - 1}
\newcommand{\DminBalancedBistarTreeFloor}{\left\lfloor \frac{1}{8}(n+2)^2 \right\rfloor - 1}
\newcommand{\DminQuasistarTreeFloor}{\left\lfloor \frac{1}{4} (n - 1)^2 \right\rfloor + 1}
\newcommand{\DminStarTreeFloor}{\left\lfloor \frac{1}{4}n^2 \right\rfloor}

\newcommand{\DminCaterpillarTreeMod}{\frac{1}{4}\left(n\DegreeSecondMoment + q\right)}
\newcommand{\DminBistarTreeMod}{\frac{1}{2}k_1(k_1 - n) + \frac{1}{4}\left[n(n + 2) + q'\right] - 1}
\newcommand{\DminBalancedBistarTreeMod}{\frac{1}{8}(n^2+4n- 4 - \phi)}
\newcommand{\DminQuasistarTreeMod}{\frac{1}{4}[n(n - 2) + {n \bmod 2}] + 1}
\newcommand{\DminStarTreeMod}{\frac{1}{4}(n^2 - {n \bmod 2})}

\newcommand{\DminBalancedBistarTreeModDecrement}{\left[3 - 4 \left(\left\lfloor \frac{n}{2} \right\rfloor \bmod 2\right)\right](1 - {n \bmod 2})}
\newcommand{\DminBalancedBistarTreeModDecrementCarlos}{((n+2)^2 \bmod 8)}

%% file: introduction.tex
A fundamental problem in network science is the normalization of the distance between vertices \cite{Latora2001a, Ferrer2004b, Barthelemy2018a, Ferrer2003a, Zamora-Lopez2019a}.
The problem is actually two-fold depending on whether the focus is on topological distance, i.e. the distance between vertices in terms of number of edges \cite{Latora2001a, Ferrer2003a, Zamora-Lopez2019a} or physical distance, i.e. the distance between vertices in some metric space, that may not be Euclidean \cite{Ferrer2004b, Barthelemy2018a}. 

Concerning topological distance, namely distance on a network, the simplest measure of topological distance is the average path length or characteristic path length, which can be defined on an undirected network $g$ of $n$ vertices as \cite{Latora2003a},
\begin{equation*}
\left< l \right>^g = \frac{1}{{n \choose 2}} \sum_{i<j} l_{ij}, 
\end{equation*}
where $l_{ij}$ is the minimum distance in edges between vertices $i$ and $j$ in $g$. 

\begin{figure}
\centering
\includegraphics{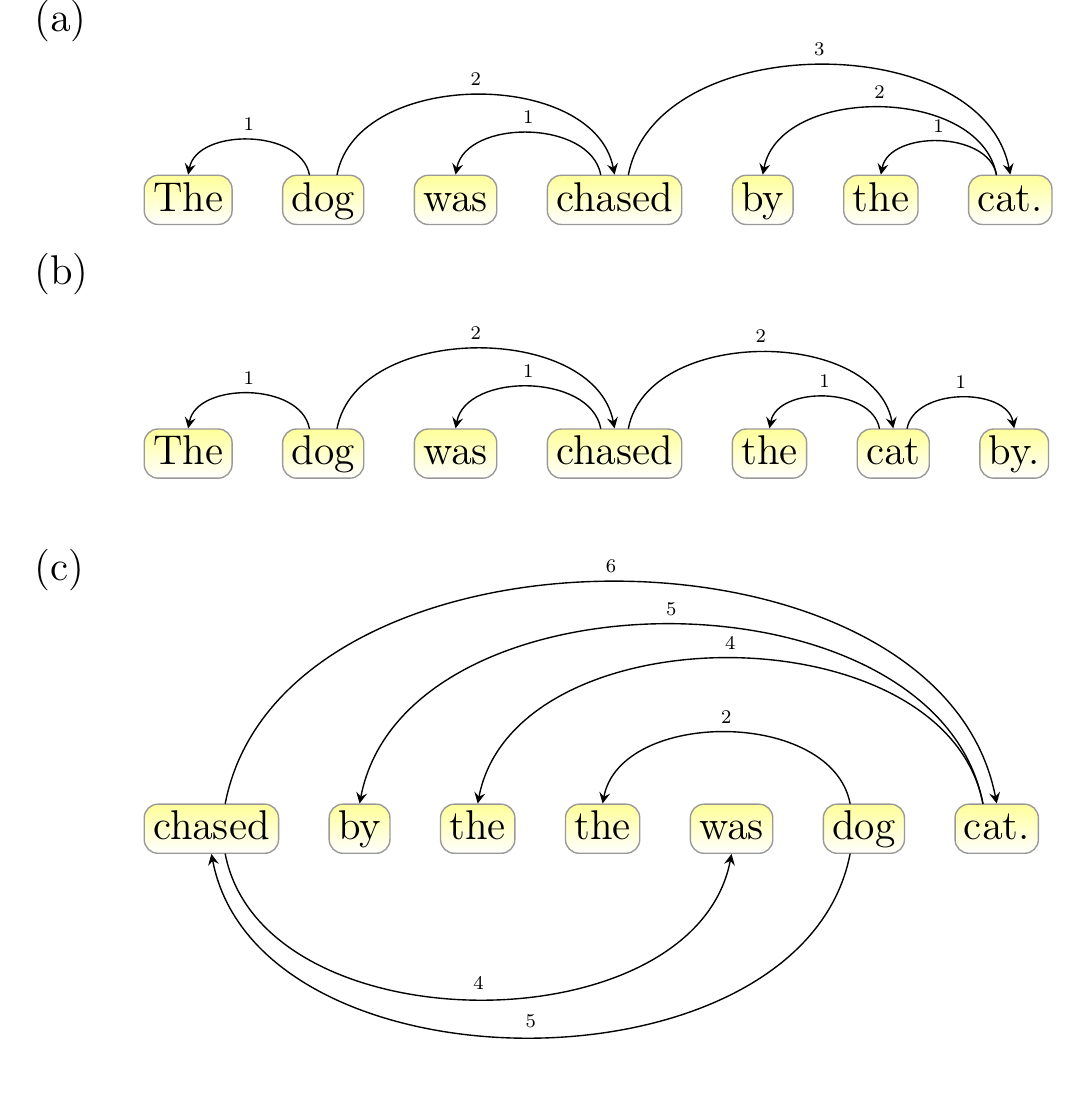} 
\caption{\label{syntactic_dependency_tree_figure} Different linear arrangements of the same syntactic dependency tree. Here edge labels indicate the distance between the linked words (in words). (a) Original linear arrangement with $D^t=10$ as it appears in \url{https://universaldependencies.org/introduction.html}. (b) Minimum linear arrangement, i.e. $D^t = D_{min}^t = 8$. $D_{min}^t = 8$ is obtained from Table \ref{summary_table} noting that $t$ is a caterpillar with $n \DegreeSecondMoment = 26$ and $q = 6$. (c) Maximum linear arrangement, i.e. $D^t = D_{max}^t = 26$. $D_{max}^t$ was computed as the maximum of $D^t$ over the $7!$ linear arrangements. } 
\end{figure}

When a graph $g$ is embedded in some space (by assigning to each vertex a position in that space), the length of a link is the physical distance between the two vertices it connects and $D^g$ represents the sum of the lengths of all links of $g$ \cite{Barthelemy2018a}. $D^g$ can be defined as
\begin{equation*}
D^g = \sum_{i<j} a_{ij}d_{ij}, 
\end{equation*} 
where $a_{ij}$ indicates if vertices $i$ and $j$ are connected ($a_{ij} = 1$ if they are connected; $a_{ij} = 0$ otherwise) and $d_{ij}$ is the physical distance between vertices $i$ and $j$ in $g$. 
When vertices $i$ and $j$ are linked, $d_{ij}$ is the edge length. 
In a network of $m$ edges, the average edge length \cite{Ferrer2004b} 
\begin{equation*} 
\left<d \right>^g = \frac{1}{m}D^g,
\end{equation*}
is the counterpart of $\left< l \right>^g$ in physical space. However, notice that $\left< l \right>^g$ is an average over all pairs of vertices, no matter if they are connected or not. In contrast, $\left<d \right>^g$ and $D^g$ are restricted to pairs of linked vertices. 

Traditionally, $D^g$ has been defined on a Euclidean two-dimensional space \cite{Barthelemy2018a}.    
Here we focus on the problem of the range of variation of the physical distance $D^g$ in one physical dimension when the network structure (namely the adjacency matrix $A=\left\{a_{ij}\right\}$ that defines $g$) remains constant. 
We consider a particular embedding in one dimension: linear arrangements of vertices whereby the position of a vertex is its position in a sequential ordering of the vertices. A prototypical example that motivates our research 	is the syntactic dependency network of a sentence, where vertices are words, edges indicate syntactic dependencies and the order of the words in the sentence defines a linear arrangement: the 1st word of the sentence takes position 1, the second word of the sentence takes position 2 and so on (figure \ref{syntactic_dependency_tree_figure}) \cite{Liu2017a}. There the distance between two vertices is usually defined as the absolute value of the difference between vertex positions: then consecutive words are at distance 1, words separated by a word are at distance 2 and so on \cite{Ferrer2004b}. In the example (figure \ref{syntactic_dependency_tree_figure} (a)), $D^g = 10$. The primary goal of this article is to establish some mathematical foundations for research on the range of variation of physical distance when $g$ is a fixed tree and physical distances are determined by the linear arrangement of its vertices as in figure \ref{syntactic_dependency_tree_figure}, where different linear arrangements of the same $g$ are shown. In particular, we aim to make a contribution on this physical distance that parallels the current understanding of the variation of topological distance \cite{Zamora-Lopez2019a} while establishing mathematical foundations for research on optimality scores on this physical distance. For these reasons, we review next the state of research on the range of variation on topological distances and their normalization.
 
In a connected network, $\left< l \right>^g$ varies 
between its value in a complete graph, a graph with as many edges as possible, and a linear tree, a tree with maximum degree two (figure \ref{various_trees_figure} (a)), i.e. \cite{Ferrer2003a}
\begin{equation*}
\left< l \right>^{complete} \leq \left< l \right> \leq \left< l \right>^{linear},
\end{equation*}
where 
\begin{eqnarray*}
\left< l \right>^{complete} = 1 \\ 
\left< l \right>^{linear} = \frac{n+1}{3}. 
\end{eqnarray*}
In trees, connected networks minimizing the number of edges, one has \cite{Ferrer2003a}
\begin{equation} 
\left< l \right>^{star} \leq \left< l \right>^g \leq \left< l \right>^{linear},
\label{variation_of_mean_topological_distance_equation}
\end{equation}
where 
\begin{equation*}
\left< l \right>^{star} = 2(n-1)/n
\end{equation*}
corresponds to a star tree, a tree with a hub of maximum degree, namely $n-1$ (figure \ref{various_trees_figure} (b)).
In \cite{Ferrer2003a}, the ratio 
\begin{equation*}
\lambda^g = \frac{\left< l \right>^g}{\left< l \right>^{linear}},
\end{equation*}
was used as a normalized measure of topological distance cost ($\lambda \leq 1$). Recently, two normalizations of $\left< l \right>$ have been investigated \cite{Zamora-Lopez2019a}
\begin{eqnarray*}
\lambda_{'}^g  & = & \frac{\left< l \right>^g}{\left< l \right>_{US}^g} \\
\lambda_{''}^g & = & \frac{\left< l \right>^g - \left< l \right>_{US}^g}{\left< l \right>_{UL}^g - \left< l \right>_{US}^g},
\end{eqnarray*} 
where 
$\left< l \right>_{US}^g$ and $\left< l \right>_{UL}^g$ are the minimum (Ultra Short) and the maximum (Ultra Long) value of $\left< l \right>$ of a network with same number of vertices and edges. In a tree,
$\left< l \right>_{US}^g = \left< l \right>^{star}$ and $\left< l \right>_{UL}^g = \left< l \right>^{linear}$.

\begin{figure}
\centering
\includegraphics{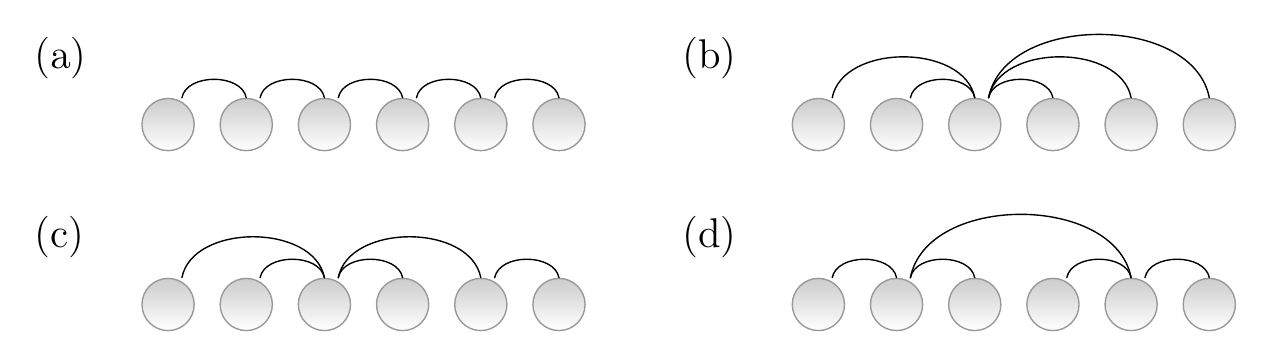} 
\caption{\label{various_trees_figure} Different trees with $n=6$ vertices and maximum degree $k_1$. (a) Linear tree ($k_1 = 2$). (b) Star tree ($k_1 = n - 1 = 5$). (c) Quasistar tree ($k_1 = n - 2 = 4$). (d) Balanced bistar tree ($k_1 = \lceil n/2 \rceil = 3$), that is formed by two star trees of 3 vertices joined by their respective hubs. }
\end{figure}

A limitation of $\left< l \right>$ is that $l_{ij} = \infty$ for vertices in different connected components and then $\left< l \right>$ is not finite in disconnected graphs, regardless of how closely connected vertices are within each component \cite{Latora2003a}. For this reason, an alternative is the so-called network {\em efficiency} \cite{Latora2001a}, an average of $1/l_{ij}$ that can be defined as 
\begin{equation*}
E^g = \frac{1}{{n \choose 2}} \sum_{i<j} \frac{1}{l_{ij}}.
\end{equation*} 
$1/E^g$ is a harmonic mean and $E$ is an average that is already normalized, since
\begin{equation*}
E^{empty} \leq E^g \leq E^{complete}.
\end{equation*}
$E^{empty} = 0$ and $E^{complete} = 1$ correspond to an empty network (a network with no edges) and a complete graph, respectively. $E^g$ is a normalized measure of optimality of a network with respect to topological distance. 

As for the situation of research on physical distance, $D^g$ has been compared against $D_{MST}^g$, the value of $D$ of a minimum spanning tree of the original graph, namely a tree $t$ defined on a subset of the edges of the original graphs such that $D^t$ is minimum \cite{Barthelemy2018a}. As $D_{MST}^g \leq D^g$, 
a normalized measure of physical distance cost is \cite{Barthelemy2018a}
\begin{equation*}
C^g = \frac{D^g}{D_{MST}^g},
\end{equation*}
namely a measure the degree of optimality of a network from the perspective of the topological distance. In that setup, the optimization problem involved in the calculation of $D_{MST}^g$ consists in keeping the number of vertices constant while allowing the network structure to vary. In contrast, we are interested in the variation of $D^g$ when the network structure remains constant, i.e. the limits of the variation of $D$ over the $n!$ linear arrangements.   

Given a network $g$, the calculation of $D_{min}^g$, the minimum value of $D^g$ over all linear arrangements is known as the minimum linear arrangement problem \cite{Diaz2002}, whereas the calculation of the maximum , i.e. $D_{max}^g$, is known as the maximum linear arrangement problem \cite{Hassin2001a}. In both problems, the network structure is fixed, as in the different linear arrangements in figure \ref{syntactic_dependency_tree_figure}.
Both problems are computationally hard \cite{Diaz2002,Hassin2001a}. In a tree $t$, the minimum linear arrangement problem simplifies and can be computed in polynomial time \cite{Shiloach1979,Esteban2015a,Chung1984} but still formulae for $D_{min}^t$ and $D_{max}^t$ are only available for specific trees \cite{Iordanskii1974a,Chung1984,Ferrer2013b}. 

A linear arrangement is planar if there are no edge crossings \cite{Frederickson88}.
Many real spatial networks arranged in two dimensions are planar or quasi planar \cite{Barthelemy2018a}. In one dimension, the concept of a planar linear arrangement has applications in areas like circuit layout \cite{Raghavan82} or dependency syntax \cite{GomNiv2013}.
In planar linear arrangements, the problem of the minimum $D^t$ simplifies further: $D_{min, P}^{t}$, the minimum value of $D$ over the all the planar ($P$) linear arrangements of a tree $t$, can be computed in linear time \cite{Iordanskii1987a, Hochberg2003a}. The first algorithm to calculate $D_{min, P}^{t}$ is due to Iordanskii \cite{Iordanskii1987a}, 16 years before Hochberg \& Stallmann's \cite{Hochberg2003a}. 

$D_{min}^t$ and $D_{max}^t$ and their limits of variation are relevant for research on the efficiency of language, where various optimality scores have been considered \cite{Ferrer2004b,Tily2010a,Gulordava2015}. The first optimality score for $D^t$ that was defined is \cite{Ferrer2004b,Tily2010a}
\begin{equation}
\Gamma^t = \frac{D^t}{D_{min}^t}.
\label{gamma_intro_equation}
\end{equation}
$\Gamma^t$ is the analog of the physical distance cost $C^g$ for research on $D^g$ where $g$ is a fixed tree $t$ and $D^t$ varies depending only on the linear arrangement.  
Another score that has been considered is \cite{Gulordava2015}
\begin{equation*}
\Delta^t = D^t - D_{min}^t.
\label{delta_intro_equation}
\end{equation*}
These limits are also relevant for a recently introduced $z$-scored value of $D$, i.e. \cite{Ferrer2018a}
\begin{equation}
D_z^t = \frac{D - D_{rla}}{(\V_{rla}^t)^{1/2}},
\label{z_scored_sum_of_edge_lengths_equation}
\end{equation}
where $D_{rla}$ and $\V_{rla}^t$ are, respectively, the expectation and the variance of $D^t$ in a uniformly random linear arrangement (r.l.a.). 
$D_{rla}$ depends only on $n$, as \cite{Ferrer2019a}
\begin{equation}
\E_{rla}[D^t] = \frac{1}{3}(n^2 - 1).
\label{sum_of_dependecy_lengths_random_equation}
\end{equation}
Table \ref{syntactic_dependecy_tree_table} shows the value of the different dependency length scores for the linear arrangements of the tree in figure \ref{syntactic_dependency_tree_figure}.
\begin{table}
\caption{\label{syntactic_dependecy_tree_table}
The statistical properties of the tree of figure \ref{syntactic_dependency_tree_figure} and its linear arrangements. The tree has $n = 7$ vertices, $n \DegreeSecondMoment = 26 $, $D_{rla} = 48/3$ (equation \ref{sum_of_dependecy_lengths_random_equation}) and  $\V_{rla}^t = 148/15$ (equation \ref{variance_of_sum_of_edge_lengths_tree_equation}). $\Delta^t$, $\Gamma^t$ and $D_z^t$ are calculated applying equations \ref{delta_intro_equation}, \ref{gamma_intro_equation} and \ref{z_scored_sum_of_edge_lengths_equation}.}
\begin{indented}
\item[]
\begin{tabular}{lrrr}
\br
           & (a) & (b) & (c) \\
\mr
$D^t$      & 10  & 8 & 20 \\
$\Delta^t$ & 2 & 0 & 10   \\
$\Gamma^t$ & $\frac{5}{4}$ & 1 & 2  \\
$D_z^t$    & $-3\sqrt{\frac{15}{37}}$ & $-4\sqrt{\frac{15}{37}}$ & $2\sqrt{\frac{15}{37}}$ \\
\br
\end{tabular}
\end{indented}
\end{table}

The remainder of the article is organized as follows. Section \ref{review_section} details the mathematical problems on $D_{min}^t$ and $D_{max}^t$ that are investigated in this article while reviewing previous results. In short, these problems are $D_{min}^t$ and $D_{max}^t$ in specific trees or classes of trees and the variation of $D_{min}^t$ and $D_{max}^t$ over all trees of the same size. Section \ref{D_max_section} investigates $D_{max}^t$. Section \ref{D_min_section} investigates $D_{min}^t$. 
Applying findings from the preceding sections, Section \ref{maxima_of_optimality_scores_section} investigates the limits of the variation of the optimality scores $\Delta^t$ and $\Gamma^t$ while Section \ref{zeta_score_section} investigates those of $D_z^t$.
Finally, Section \ref{discussion_section} reviews all our findings and suggests future research problems.

%% file: review.tex
\begin{table}
\caption{\label{hubiness_table} The distribution of vertex degrees on a tree $t$ of $n$ vertices, for specific trees. $k_1$ is the maximum vertex degree, and $\DegreeSecondMoment$, the second moment of degree about zero. $\DegreeSecondMoment$ for linear and star trees is borrowed from \cite{Ferrer2013b}. $\DegreeSecondMoment$ for quasistar trees is borrowed from \cite{Ferrer2014f} and that of balanced bistar trees is derived from equation \ref{degree_2nd_moment_bistar_equation} with $k_1 = \lceil n/2 \rceil$. }

\begin{indented}
\item[]
\begin{tabular}{lrr}
\br
$t$       & $k_1$ & $\DegreeSecondMoment$ \\
\mr
linear    & 2     & $\DegreeSecondMomentLinearTree$ \\
 & & \\
balanced bistar & $\left\lceil \frac{n}{2} \right\rceil$ & $\frac{2}{n}(\lceil n/2 \rceil(\lceil n/2 \rceil - n) - 1) + n + 1$ \\  
 & & \\
quasistar & $n - 2$ & $\DegreeSecondMomentQuasistarTree$ \\
 & & \\
star      & $n - 1$ & $\DegreeSecondMomentStarTree$ \\
 & & \\
\br
\end{tabular}
\end{indented}
\end{table}

Given the potential to obtain simple formulae for trees and the interest of trees in language research \cite{Liu2017a,Temperley2018a}, here we are interested in three kinds of problems over trees of $n$ vertices.

\subsection{$D_{min}^t$ and $D_{max}^t$ in specific trees}
\label{problem1}

We investigate $D_{min}^t$ and $D_{max}^t$ in specific kinds of trees (distinct unlabelled trees) that are selected for their theoretical interest. Linear trees and star trees are relevant to understand the variation of topological distance as we have seen above (equation \ref{variation_of_mean_topological_distance_equation}) \cite{Ferrer2003a} and also to understand the limits of the variation of $D_{min}^t$ \cite{Iordanskii1974a,Esteban2016a}. Trivially \cite{Chung1984, Ferrer2013b, Esteban2016a}, 
\begin{equation*}
D_{min}^{linear} = n - 1.
\end{equation*}
Iordanskii found that \cite{Iordanskii1974a}, 
\begin{equation*}
D_{min}^{star} = \DminStarTreeFloor,
\end{equation*}
which was rediscovered later in equivalent forms \cite{Ferrer2013b, Esteban2016a}, e.g. 
\begin{equation*}
D_{min}^{star} = \DminStarTreeMod.
\end{equation*}
Bistar trees ({\em bistar}) consist of two stars joined by the hub and include star trees 
as an extreme case when one of the stars has only one vertex (figures \ref{various_trees_figure} (b-d)) \cite{San_Diego2014a, Vaidya2018a}. Here we are interested in two distinct representatives: quasistar trees ({\em quasi}), where one of the original stars has only two vertices (figure \ref{various_trees_figure} (c)) and balanced bistar trees ({\em b-bistar}), where the two original stars have the same size or differ in one vertex (figure \ref{various_trees_figure} (d)). 
Quasistar trees are important for the theory of edge crossings in linear arrangements \cite{Ferrer2014f,Alemany2018a}.
In this article, we will unveil that balanced bistar trees maximize $D_{max}^t$ over trees of $n$ vertices.
We will also obtain formulae for $D_{min}^{quasi}$ and $D_{min}^{b-bistar}$.

It has been shown that \cite{Ferrer2013b}
\begin{equation*}
D_{max}^{star} = \DmaxStarTree.
\end{equation*} 
$D_{max}^{linear}$ is unknown but $D_{max, P}^{t}$, the maximum value of $D$ over the all the planar ($P$) linear arrangements of a tree $t$, has been investigated. 
It has been shown that \cite{Ferrer2013b}
\begin{equation*}
D_{max}^{star} = D_{max,P}^{linear} = \DmaxStarTree.
\end{equation*}
Notice that edge crossings are impossible in a star tree \cite{Ferrer2013b} and hence $D_{max}^{star} = D_{max,P}^{star}$.
Here we will calculate $D_{max}^{linear}$ as well as $D_{max}^{quasi}$ and $D_{max}^{b-bistar}$.  

\subsection{$D_{min}^t$ and $D_{max}^t$ in classes of trees}
\label{problem2}

We investigate $D_{min}^t$ and $D_{max}^t$ in classes of trees (comprising more than one distinct unlabelled tree but not all distinct labelled trees). Two classes are selected for their theoretical interest: bistar trees (for the reasons explained above) and caterpillar trees ({\em cat}). Caterpillar trees is the class of trees such that when all the leaves are removed a linear tree is left \cite{Rosen2017a}. Caterpillar trees are relevant for being a generalization of linear trees and bistar trees of enough simplicity that simple formulae for $D_{min}^t$ can be obtained \cite{Horton1997a}.  
For each relevant class, we aim to express $D_{min}^t$ and $D_{max}^t$ as a function of $n$ and additional parameters of the networks extracted from vertex degrees: e.g., $k_1$, the largest degree, or $\DegreeSecondMoment$, the second moment of degree about zero. 

\subsection{The variation of $D_{min}^t$ and $D_{max}^t$ over all trees of the same size.}
\label{problem3}

We investigate the variation of $D_{min}^t$ and $D_{max}^t$ over all distinct unlabelled trees of $n$ vertices. The problem is motivated by research on $D^t$ as a function of $n$ \cite{Ferrer2004b,Park2009a,Ferrer2013c,Futrell2015a}.
It is well-known that any tree $t$ of $n$ vertices satisfies \cite{Esteban2016a}
\begin{equation}
D_{min}^{linear} \leq D_{min}^t \leq D_{min}^{star} \leq D_{rla}.
\label{Dmin_chain_equation}
\end{equation} 
The part $D_{min}^t \leq D_{min}^{star}$ is due to Iordanskii \cite{Iordanskii1974a} although rediscovered later \cite{Esteban2016a}. 
Asymptotically ($n \rightarrow \infty$), one also has that \cite{Iordanskii1987a}, 
\begin{equation*}
D_{min,P}^t \leq \frac{3}{2}D_{min}^t.
\end{equation*}
An inequality equivalent to equation \ref{Dmin_chain_equation} for $D_{max}^t$ is not forthcoming but it has been shown that any tree $t$ of $n$ vertices satisfies \cite{Ferrer2013b}
\begin{equation*}
D_{max,P}^t \leq D_{max,P}^{linear} = D_{max}^{star}.
\end{equation*}
Here we will show that any tree $t$ of $n$ vertices also satisfies
\begin{equation}
D_{rla} \leq D_{max}^{star} \leq D_{max}^t \leq D_{max}^{b-bistar}. 
\label{Dmax_chain_equation}
\end{equation}

The hubiness of a tree is defined by $\DegreeSecondMoment$, the second moment of degree about zero \cite{Ferrer2013b,Ferrer2017a} (Table \ref{hubiness_table}). $\DegreeSecondMoment$ and $k_1$, the maximum vertex degree are closely related for sufficiently large $k_1$.
Table \ref{summary_table} summarizes all the existing results and the new results that are presented in this article for the problems defined in Sections \ref{problem1} and \ref{problem2}.

\begin{table}
\caption{\label{summary_table} 
$D_{min}^t$ and $D_{max}^t$, the minimum and the maximum value of $D$, the sum of edge lengths of a tree $t$ for different classes of trees. Classes are sorted from more general to more concrete. Classes formed by a single tree are sorted increasingly by their hubiness (Table \ref{hubiness_table}). $n$ is the number of vertices of the tree and $k_1$ is the largest degree, $q$ is the number of vertices of odd degree, $q' = {k_1 \bmod 2} + {(n - k_1) \bmod 2}$ and $\phi = \DminBalancedBistarTreeModDecrementCarlos$. For $D_{min}^t$ we provide at least two formulae: one based on the floor or ceil function and the other based on ${\bmod}$ (the exception are linear trees due to the simplicity of their formula). Formulae without a reference attached are new to our knowledge. }	

\begin{indented}
\item[]
\begin{tabular}{lrr}
\br
$t$       & $D_{min}^t$ & $D_{max}^t$ \\
\mr
 & & \\
caterpillar & $\DminCaterpillarTreeFloorHorton$ \cite{Horton1997a} & \\
            & $\DminCaterpillarTreeFloor$ & \\ 
            & $\DminCaterpillarTreeMod$ & \\
 & & \\
bistar    & $\DminBistarTreeFloor$ & $\DmaxBistarTree$ \\
          & $\DminBistarTreeMod$ & \\ 
 & & \\
linear    & $\DminLinearTree$ \cite{Chung1984} & $\DmaxLinearTreeFloor$ \\
          &                   & $\DmaxLinearTreeMod$ \\
 & & \\
balanced bistar & $\DminBalancedBistarTreeFloor$ & $\DmaxBalancedBistarTree$ \\
                & $\DminBalancedBistarTreeMod$ & \\ 
 & & \\
quasistar & $\DminQuasistarTreeFloor$ & $\DmaxQuasistarTree$ \\
          & $\DminQuasistarTreeMod$ & \\
 & & \\
star      & $\DminStarTreeFloor$ \cite{Iordanskii1974a} & $\DmaxStarTree$ \cite{Ferrer2013b} \\
          & $\DminStarTreeMod$ \cite{Esteban2016a} & \\
 & & \\
\br
\end{tabular}
\end{indented}
\end{table}

%% file: maximum_sum_of_edge_lengths.tex
\label{D_max_section}

Here we investigate $D_{max}^t$ in linear trees and bistar trees as well as the limits of the variation of $D_{max}^t$ over all trees of $n$ vertices. 

\begin{figure}
\centering
\includegraphics[scale = 0.8]{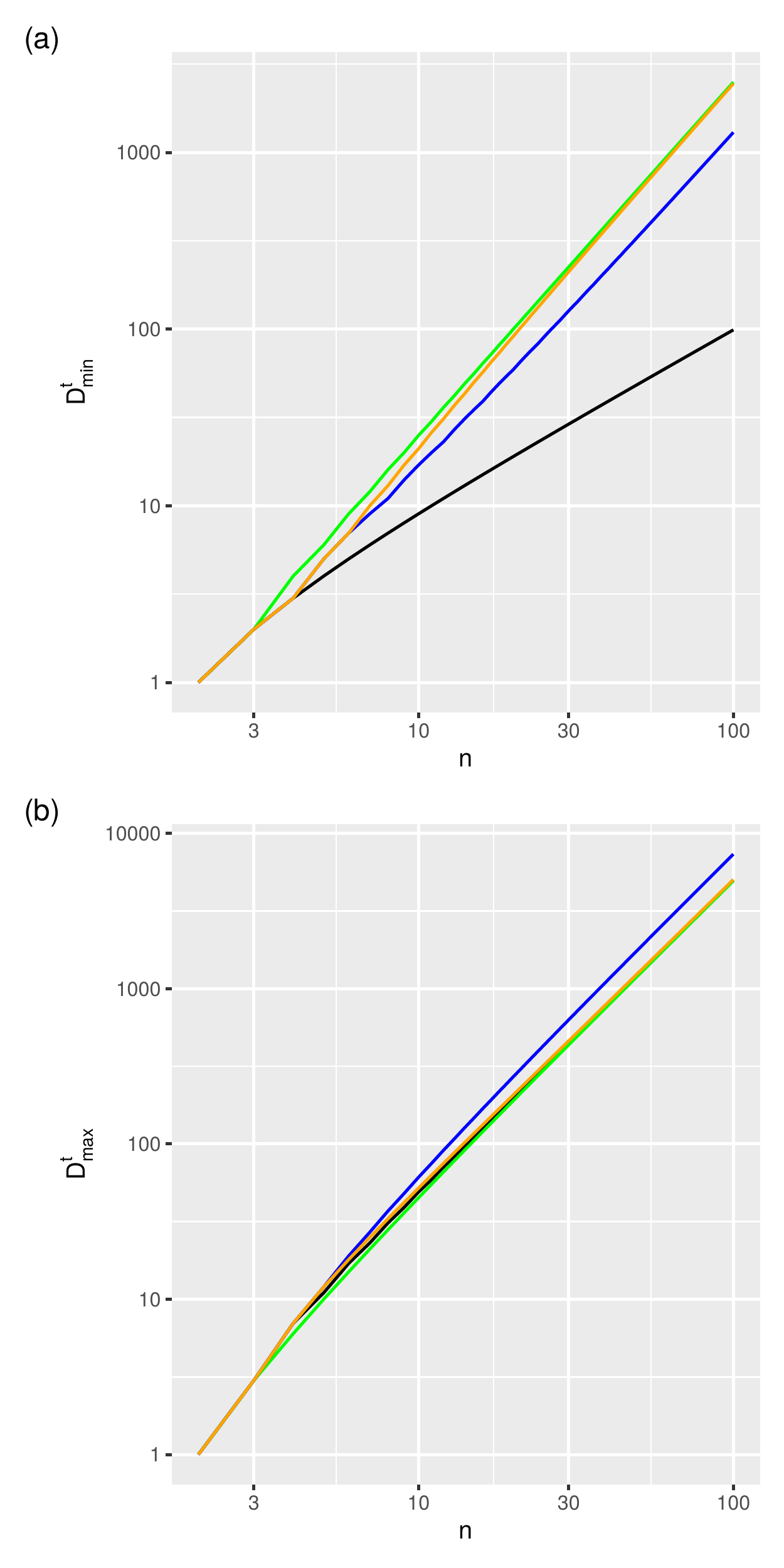}
\caption{\label{scaling_of_D_min_and_D_max_figure} The scaling of the limits of the variation of $D^t$ as a function of $n$, the number of vertices of the tree $t$, for different trees: linear trees (black), balanced bistar trees (blue), quasistar trees (orange) and star trees (green). (a) $D_{min}^{t}$. (b) $D_{max}^{t}$. }
\end{figure}

\subsection{Linear trees}



A linear tree is a tree whose vertices are linked as a chain, i.e., a tree with arcs of the form $\{v_1,v_2\}, \{v_2,v_3\}, \ldots, \{v_{n-1},v_n\}$.
The maximum value of $D^{linear}$ over the $n!$ linear arrangements is (\ref{D_max_linear_trees_appendix})
\begin{eqnarray}
D_{max}^{linear} & = & \DmaxLinearTreeMod \nonumber \\ 
                 & = & \DmaxLinearTreeFloor. \label{maximum_linear_arrangement_of_linear_tree_equation}
\end{eqnarray}

\subsection{Bistar trees}

Hereafter we assume that a vertex is labelled with its position in the degree sequence, namely the non-increasing sequence of vertex degrees. Then $k_i$ is the degree of the vertex with the $i$-th largest degree.  
A bistar tree is a generalization of trees of high theoretical interest: star trees \cite{Iordanskii1974a,Esteban2016a} and quasistar trees \cite{Ferrer2014f,Alemany2018a}. If $k_1 = n-1$ (hence $k_2 = 1$) then we have a star tree. If $k_1 = n - 2$ (hence $k_2 = 2$) then we have a quasistar tree (figures \ref{various_trees_figure} (b-d)). Since a bistar tree consists of two joined stars, one may think that a bistar tree has three parameters, $n$, $k_1$ and $k_2$. However, $n$ and $k_1$ suffice, as we will see next.

A bistar tree with $n \geq 2$ vertices satisfies the following properties:
\begin{enumerate}
\item
It has at most two internal vertices, more precisely 
$2 - \delta_{k_1, 1} - \delta_{k_2, 1}$
internal vertices, where $\delta$ is the Kronecker delta function. $\delta_{k_2, 1} = 1$ when the tree is a star.	
\item
It has $n - 2 + \delta_{k_1, 1} + \delta_{k_2, 1}$ leaves.
\item
Then $k_i = 1$ for $3 - \delta_{k_1, 1} - \delta_{k_2, 1} \leq i \leq n$.
\item
\begin{equation}
k_2 = n - k_1
\label{2nd_largest_degree_equation}
\end{equation}
because the sum of vertex degrees must satisfy
\begin{equation*} 
k_1 + k_2 + n - 2 = 2(n-1)
\end{equation*}
by the handshaking lemma. 
\item
\begin{eqnarray}
\DegreeSecondMoment & = & \frac{1}{n}(k_1^2 + (n-k_1)^2 + n - 2) \nonumber \\
                   & = & \frac{2}{n}(k_1(k_1 - n) - 1) + n + 1 \label{degree_2nd_moment_bistar_equation}
\end{eqnarray}
\item
\begin{equation}
k_1 \geq \LargestDegreeBalancedBistarTree.
\label{lower_bound_of_largest_degree_equation}
\end{equation} 
Combining equation \ref{2nd_largest_degree_equation} with the condition $k_1 \geq k_2$, one obtains
\begin{equation*}
k_1 \geq \frac{n}{2},
\end{equation*}
which knowing that $k_1$ is an integer gives equation \ref{lower_bound_of_largest_degree_equation}.
\end{enumerate}
Our definition of a star tree with two parameters, $n$ and $k_1$, is equivalent to other two-parameter definitions. \cite{San_Diego2014a} defines a bistar with two parameters $n_1$ and $n_2$. The bistar is formed by taking a graph with a single edge and two vertices and adding $n_1$ edges at one end of the edge and $n_2$ edges at the other end. Ours is then $n = 2 + n_1 + n_2$ and $k_1 = \max(n_1, n_2) + 1$.   
\cite{Vaidya2018a} defines a bistar with two parameters $n_1'$ and $n_2'$. The bistar is formed by adding an edge between the hubs of two stars of $n_1'$ and $n_2'$ vertices respectively \cite{Vaidya2018a}. Ours is then $n = n_1' + n_2'$ and $k_1 = \max(n_1', n_2')$.   
The term bistar tree has also been used to refer to a tree with only one inner edge or a tree of diameter three, where the diameter is the length of the longest shortest path in edges \cite{Dragan2006a}. This is not exactly our definition of bistar because it excludes star trees and implies $n \geq 4$. In our definition, a bistar tree has at most one inner edge and diameter at most 3 and is then valid for $n < 4$. 

We introduce a bistar tree of great theoretical importance to calculate the maximum of $D_{max}^t$ over all trees of same size: the balanced bistar tree (figure \ref{bistar_figure}). That tree is a bistar tree with 
\begin{equation} 
k_1 = \LargestDegreeBalancedBistarTree.
\label{maximum_degree_balanced_bistartree_equation}
\end{equation}
The latter implies that $k_2 = \SecondLargestDegreeBalancedBistarTree$ thanks to equation \ref{2nd_largest_degree_equation}.
The term balanced comes from the fact a balanced bistar tree is a bistar tree where the difference $k_1 - k_2$ is minimized. Thanks to equation \ref{2nd_largest_degree_equation}, one has that
\begin{equation*}
k_1 - k_2 = 2 k_1 - n.
\end{equation*}
The fact that $k_1 \geq k_2$, gives that the difference is minimized when $k_1$ satisfies equation \ref{maximum_degree_balanced_bistartree_equation}. 

  
In a bistar tree $t$, the maximum value of $D_{t}$ over all $n!$ linear arrangements is (\ref{D_max_bistar_tree_appendix}) 
\begin{equation*}
D_{max}^{bistar} = k_1 (n - k_1) + \frac{n}{2}(n - 3) + 1.
\end{equation*}
Table \ref{summary_table} summarizes $D_{max}^{t}$ for specific bistar trees: balanced bistar trees, quasistar trees and star trees (the formulae are derived in \ref{D_max_bistar_tree_appendix}). 

\subsection{The maximum $D_{max}^t$}

In a graph $g$ of $n$ vertices and $m$ edges, an obvious upper bound of $D_{max}^t$ is \cite{Ferrer2018a}
\begin{equation*} 
D_{upper, naive}^{g} = m(n-1),
\end{equation*}
where $n-1$ is the maximum length of an edge. 
{\em A priori}, $n - d$ edges of length $d$ can be formed. Taking $m$ lengths as long as possible, the analog of Petit's edge method (EM) for the maximum linear arrangement problem \cite{Petit2003a}, gave another upper bound of $D_{max}^t$ \cite{Ferrer2018a} that is
\begin{eqnarray*}
D_{upper,EM}^g 
   & = & (m - F(d_*))(d_* - 1) \\
   &   & + \frac{1}{6}(n- d_*) (n^2 + (n+3)d_* - 2d_*^2 - 1),
\end{eqnarray*}
where
\begin{eqnarray*}
F(d_0) & = & \frac{1}{2}(n - d_0)(n - d_0 + 1) \\ 
d_*    & = & \left\lceil n + \frac{1}{2} - \frac{1}{2}\sqrt{8m + 1} \right\rceil. 
\end{eqnarray*}
Figure \ref{balanced_bistar_versus_Petits_methods_figure} shows that, when $m = n - 1$ as in a tree, the naive upper bound, i.e.  $(n-1)^2$, beats the edge method upper bound for sufficiently large $n$. This is likely to be due to the tree constraint (acyclicity and connectedness). Interestingly, the naive upper bound is close to the true maximum of $D_{max}^t$, that is achieved by a     
maximum linear arrangement of a balanced bistar tree as we will show next.

\begin{thm}[Maximum $D_{max}^t$]
\label{maximum_D_max_theorem}
For any tree $t$ of $n$ vertices,
\begin{equation*}
D_{max}^t \leq D_{max}^{b-bistar} = \DmaxBalancedBistarTree.
\end{equation*}
\end{thm}

\begin{proof}
Let $\tau$ be the set of all unlabelled trees of $n$ vertices.
Let $\upsilon$ be the set of labelled trees of $n$ vertices, i.e., the set of trees of $n$ vertices where each vertex has been assigned a unique number in $\{1, 2, \ldots, n\}$ that indicates its position in the linear arrangement. 
Given an unlabelled tree $t \in \tau$, choosing a linear arrangement for it (by assigning a linear order to its vertices) results into one of the trees in $\upsilon$. Thus, maximizing the value of $D^t$ across the $n!$ possible linear arrangements of each unlabelled tree in $\tau$ reduces to maximizing the value of $D$ in $\upsilon$.\footnote{Note that the mapping from linear arrangements to trees in $\upsilon$ is not bijective (different linear arrangements can result into the same labelled tree, e.g. all the linear arrangements of a star tree where the central vertex's position is kept constant) but this is not relevant for this proof, as it does not affect $D$.}

Let $\varphi_1$ be the set of directed rooted trees obtained by rooting each tree in $\upsilon$ at its vertex $1$, the 1st vertex in the linear arrangement, and directing all edges to point away from the root. Trivially, this mapping between $\varphi_1$ and $\upsilon$ is bijective (as said orientation is unique) and it preserves the sum of edge lengths. Therefore, if we find a directed tree with maximum sum of arc lengths in $\varphi_1$, its underlying undirected tree will have the maximum sum of edge lengths in $\upsilon$.

We will show that the directed tree with arcs $n \rightarrow 2, n \rightarrow 3, \ldots, n \rightarrow \lfloor \frac{n}{2} \rfloor, 1 \rightarrow \lfloor \frac{n}{2} \rfloor +1, \ldots, 1 \rightarrow n-1, 1 \rightarrow n$, whose underlying undirected tree is a balanced bistar tree as in figure \ref{bistar_figure}, maximizes the sum of arc lengths in $\varphi_1$. To see this, we use the property of directed trees that every vertex has exactly one incoming arc, except for the root which has none. Thus, for any tree of $\varphi_1$, we can write its arcs as $A_2, A_3, \ldots, A_n$ such that $A_i$ is the arc going into vertex $i$. Now, if we consider each arc individually, we can say that

\begin{itemize}
\item The length of the arc $A_i$, for $2 \le i \le \lfloor \frac{n}{2} \rfloor$, is at most $n-i$, as $n$ is the farthest possible vertex from vertex $i$. That is, the arc $n \rightarrow i$ is the longest possible arc to vertex $i$.
\item The length of the arc $A_i$, for $\lfloor \frac{n}{2} \rfloor < i \le n$, is at most $i-1$, as $1$ is the farthest possible vertex from vertex $i$. That is, the arc $1 \rightarrow i$ is the longest possible arc to vertex $i$.
\end{itemize}

The directed tree mentioned above has exactly the longest arcs going into each vertex $i$, for $2 \le i \le n$. Thus, it maximizes the sum of arc lengths in $\varphi_1$, as it maximizes each individual length $A_2, \ldots, A_n$. Therefore, the underlying balanced bistar tree has the maximum sum of edge lengths in $\upsilon$, proving the theorem. Note: the ordering of the vertices implied by such directed tree corresponds to an extreme linear arrangement as described above.  
\end{proof}

\begin{figure}
\centering
\includegraphics[scale = 0.7]{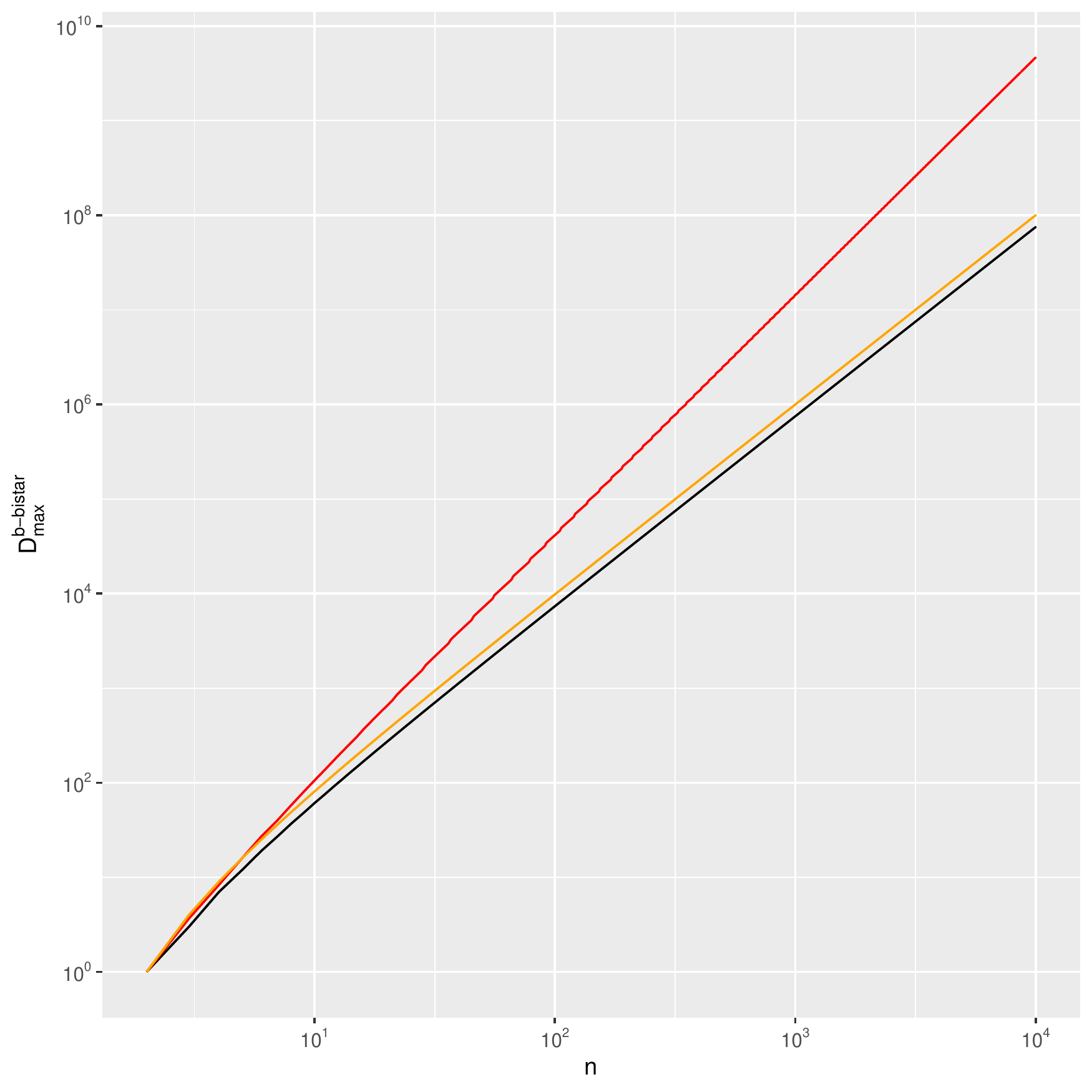} 
\caption{\label{balanced_bistar_versus_Petits_methods_figure} $D_{max}^{b-bistar}$, the true maximum of $D_{max}^t$, as a function of $n$ (black). The predictions of the naive upper bound (orange) and the edge method upper bound (red) are also shown. }
\end{figure}

The problem of maximizing $D_{max}^t$ is equivalent to finding the maximum spanning tree on a complete graph where the weight of every edge is the distance between the vertices that form it in the linear arrangement (\ref{D_max_appendix}). The argument provides an alternative way to demonstrate Theorem \ref{maximum_D_max_theorem}. 

\subsection{Lower bounds of $D_{max}^t$}

A star tree yields the minimum value of $D_{max}^t$, namely $D_{max}^t \ge D_{max}^{star}$ (\ref{minimum_D_max_appendix}). In turn,
by definition of average and maximum, $D_{rla} \leq D_{max}^t$, and particularizing this for a star tree, $D_{rla} \leq D_{max}^{star}$. Putting these results together, we have that $D_{rla} \leq D_{max}^{star} \leq D_{max}^t$. 





Figure \ref{scaling_of_D_min_and_D_max_figure} (b) shows the growth of $D_{max}^t$ for different trees.

%% file: minimum_sum_of_edge_lengths.tex
\label{D_min_section}

$D_{min}^{t}$ is the minimum value of $D$ over the all the linear arrangements of a tree $t$. $D_{min}^t$ can be calculated with rather complex algorithms for any tree $t$ \cite{Chung1984,Shiloach1979,Esteban2015a}. Algorithms to calculate $D_{min}^{t}$ satisfying a certain constraint are also available but less known. See \cite{Iordanskii1987a,Hochberg2003a} for planarity (no edge crossings) and \cite{Gildea2007a,Park2009a} for projectivity, a particular case of planarity. Here we are interested in compact formulae for $D_{min}^t$ for certain classes of trees or general lower bounds.

In his pioneering research, Iordanskii, in addition to showing that $D_{min}^t \leq D_{min}^{star}$ \cite{Iordanskii1974a}, also showed that the maximum value of $D_{min}^t$ over all trees of $n$ vertices with maximum degree $k_1$ such that $k_1 \geq 3$ satisfies the following inequalities \cite{Iordanskii1975a}
\begin{equation*}
\frac{k_1- \frac{7}{3}}{12 \log_2 k_1}n\log_2 n\left(1 - \frac{\log_2 \frac{3}{2}}{\log_2 n} \right) < \max_{t} D_{min}^t < \frac{k_1 + 4}{4\log_2(k_1-1)}n \log_2 n.
\end{equation*}
The asymptotic order in $n$ of the expression for the lower and upper bound is
\begin{equation}
\frac{k_1 n \log_2 n}{\log_2 k_1}
\end{equation}  
with a constant factor
\begin{equation*}
c_1 = \left\{
         \begin{array}{ll}
         \frac{1}{54} \approx 0.0185 & \mbox{for~} k_1 \geq 3 \\
         \frac{1}{12} \approx 0.0833             & \mbox{for~} k_1 \rightarrow \infty  
         \end{array}  
      \right.
\end{equation*}
for the lower bound 
and a constant factor
\begin{equation*}
c_2 = \left\{
         \begin{array}{ll}
         \frac{7}{12} \log_2 3 \approx 0.9244 & \mbox{for~} k_1 \geq 3 \\
         0.25                                 & \mbox{for~} k_1 \rightarrow \infty  
         \end{array}  
      \right.
\end{equation*}
for the upper bound. 

\cite{Petit2003a} reviews various techniques to obtain lower bounds of $D^t$. In a network with $m$ edges, the edge method consists of picking the $m$ shortest edges noting that there can be at most $n - d$ edges of length $d$, for $1 \leq d \leq n-1$. In a tree, this methods trivially gives $D_{min} \geq D_{min}^{linear} = n - 1$. The next theorem presents a lower bound of $D_{min}^t$ that depends exclusively on the degree sequence and that is obtained with the degree method. A similar application of the degree method can be found in \cite{Ferrer2013b}.

\begin{thm}
For any tree $t$ of $n$ vertices,
\begin{equation*}
D_{min}^t \geq D_0^t = \frac{1}{4} \left(\frac{n}{2} \DegreeSecondMoment + 2(n-1) + \frac{1}{2}q\right)
\end{equation*}
where
\begin{equation*}
q = \sum_{i=1}^n (k_i \bmod 2)
\end{equation*}
is the number of vertices of odd degree.
\end{thm}

\begin{proof}
Let $D_i^t$ be the sum of the length of the edges attached to the $i$-th vertex of $t$.
The degree method is based on a star tree decomposition of $D$ in a network, whereby \cite{Petit2003a}
\begin{equation}
D^t = \frac{1}{2} \sum_{i=1}^n D_i^t.
\label{star_decomposition_equation}
\end{equation}
$D_{i,min}^t$ is a lower bound of $D_i^t$ defined as \cite{Petit2003a}
\begin{equation*}
D_{i,min}^t = 2 \sum_{j=1}^{k_i/2} j = \frac{1}{2}\left(\frac{k_i^2}{2} + k_i\right)
\end{equation*}
if $k_i$ is even,
and
\begin{equation*}
D_{i,min}^t = \frac{1}{2}(k_i + 1) + 2\sum_{j=1}^{(k_i - 1)/2} j = \frac{1}{2}\left(\frac{k_i^2}{2} + k_i + \frac{1}{2}\right)
\end{equation*}
if $k_i$ is odd.
Combining both results, one obtains
\begin{equation*}
D_{i, min}^t = \frac{1}{2}\left(\frac{k_i^2}{2} + k_i + \frac{1}{2}(k_i \bmod 2)\right).
\end{equation*}
Inserting the last result into equation \ref{star_decomposition_equation}, one obtains
\begin{eqnarray*}
D_{min}^t & \geq & \frac{1}{2} \sum_{i=1}^n D_{i,min} \\
          & =    & \frac{1}{4} \sum_{i=1}^n \left(\frac{k_i^2}{2} + k_i + \frac{1}{2}(k_i \bmod 2)\right) \\
          & =    & \frac{1}{4} \left(\frac{n}{2} \DegreeSecondMoment + 2(n-1) + \frac{1}{2}q\right) \\
          & =    & D_0^t.
\end{eqnarray*}
\end{proof}

In a linear tree, $\DegreeSecondMoment = 4 - 6/n$ \cite{Ferrer2013b} and $q = 2$ give $D_0 = n - 1$, matching $D_{min}^{linear}$. In contrast, for a star tree, $\DegreeSecondMoment = n - 1$ \cite{Ferrer2013b} and $q = n - 1 + {(n - 1) \bmod 2}$ give 
\begin{equation}
D_0 = \frac{1}{8}\left(n^2 + 4n - 5 + {(n-1) \bmod 2}\right) \label{D_min_lower_bound_degree_method_equation}
\end{equation}
while $D_{min}^{star} = \DminStarTreeMod$ \cite{Iordanskii1974a}. Asymptotically, $D_0$ deviates from the true minimum, $D_{min}^{star}$, by a factor of $1/2$. 

\subsection{Caterpillar trees}

The following theorem is a formalization of the arguments of \cite{Horton1997a} that presents a lower bound of $D_{min}^t$ that has no deviation if the tree is a caterpillar ({\em cat}), including then the particular cases of star trees and linear trees discussed above. 
\begin{thm}[Horton \cite{Horton1997a}]
\label{lower_bound_of_D_min_theorem}
For any tree $t$ of $n$ vertices, 
\begin{equation}
D_{min}^t \geq D_{min}^{cat}
\label{lower_bound_of_D_min_equation}
\end{equation}
where $D_{min}^{cat}$ is the value of $D_{min}^t$ of a caterpillar tree with the same degree sequence as $t$.
We have 
\begin{equation}
D_{min}^{cat} = n - 1 + \sum_{i=1}^n a_i,
\label{D_min_caterpillar_equation}
\end{equation}
where
\begin{eqnarray}
a_i & = & \left\lfloor \frac{k_i}{2} \right\rfloor \left\lceil \frac{k_i - 2}{2} \right\rceil \nonumber \\
    & = & \left\lfloor \frac{(k_i - 1)^2}{4} \right\rfloor \label{Hortons_formula_equation}
\end{eqnarray}
and $k_i$ is the degree of the $i$-th vertex.
\end{thm}
\begin{proof}
In trees, the bipartite crossing number is \cite{Shahrokhi2001a}
\begin{equation}
bcr = D_{min} - n + 1 - \sum_{i=1}^n a_i.
\label{bcr_equation}
\end{equation}
$bcr \geq 0$ by definition and $bcr = 0$ if and only if the tree is a caterpillar tree \cite{Chimani2018a}. 
Therefore, equation \ref{bcr_equation} becomes equation \ref{lower_bound_of_D_min_equation} with equality if and only if the tree is a caterpillar. 
A longer proof where $a_i$ is expressed as in equation \ref{Hortons_formula_equation} is found in \cite{Horton1997a}. 
\end{proof}
The following theorem introduces useful algebraic expressions for caterpillar trees, alternating floor with modulo operations. 
\begin{thm}

\label{D_min_caterpillar_theorem}
\begin{eqnarray}
D_{min}^{cat} & = & \DminCaterpillarTreeFloorHorton \label{D_min_caterpillar_floor_Horton_equation} \\
              & = & \DminCaterpillarTreeFloor \label{D_min_caterpillar_floor_equation} \\
              & = & \DminCaterpillarTreeMod. \label{D_min_caterpillar_degree_2nd_moment_equation}
\end{eqnarray}
\end{thm}
\begin{proof}
Equation \ref{D_min_caterpillar_floor_Horton_equation} is due to \cite{Horton1997a}. 
Theorem \ref{lower_bound_of_D_min_theorem} on a star tree ($k_1 = n - 1$ and $k_i = 1$ for $2 \leq i \leq n$) gives
\begin{eqnarray*}
\left\lfloor \frac{n-1}{2} \right\rfloor \left\lceil \frac{n-1}{2} - 1 \right\rceil & = & D_{star}^{min} - (n - 1)
\end{eqnarray*}
The formulae for $D_{star}^{min}$ in Table \ref{summary_table} then give
\begin{eqnarray*}
\left\lfloor \frac{n-1}{2} \right\rfloor \left\lceil \frac{n-1}{2} - 1 \right\rceil                                                                        & = & \frac{n^2 - {n \bmod 2}}{4} - (n - 1) \\
                                                                                    & = & \left\lfloor \frac{1}{4} n^2 \right\rfloor - (n - 1).
\end{eqnarray*}
The change of variable $k_i = n - 1$ gives
\begin{eqnarray*}
a_i & = & \frac{(k_i+1)^2 - {(k_i+1) \bmod 2}}{4} - k_i \\
    & = & \left\lfloor \frac{1}{4} (k_i + 1)^2 \right\rfloor - k_i. 
\end{eqnarray*}
Plugging these results into equation \ref{D_min_caterpillar_equation}, one obtains 
\begin{eqnarray*}
D_{min}^{cat} & = & n - 1 + \sum_{i=1}^n a_i \\
              & = & n - 1 + \sum_{i=1}^n \left(\left\lfloor \frac{1}{4} (k_i + 1)^2 \right\rfloor - k_i \right) \\ 
              & = & n - 1 + \sum_{i=1}^n \left\lfloor \frac{1}{4} (k_i + 1)^2 \right\rfloor - \sum_{i=1}^n k_i \\
              & = & \DminCaterpillarTreeFloor  
\end{eqnarray*}
and also 
\begin{eqnarray*}
D_{min}^{cat} & = & n - 1 + \sum_{i=1}^n a_i \\
              & = & n - 1 + \frac{1}{4}\sum_{i=1}^n \left[k_i^2 + 2 k_i + 1 - {(k_i + 1) \bmod 2}\right] - \sum_{i=1}^n k_i \\ 
              & = & \frac{1}{4}\sum_{i=1}^n \left[k_i^2 + {k_i \bmod 2}\right]
 + \frac{1}{2} \sum_{i=1}^n k_i - (n - 1) \\ 
              & = & \DminCaterpillarTreeMod. 
\end{eqnarray*}
\end{proof}

It is easy to see that $D_{min}^{cat}$ is a tighter lower bound of $D_{min}^t$ than $D_0^t$. Thanks to equations \ref{D_min_lower_bound_degree_method_equation} and \ref{D_min_caterpillar_degree_2nd_moment_equation}, the condition $D_{min}^{cat} \geq D_0^{t}$ is equivalent to 
\begin{equation*}
\DegreeSecondMoment \geq 4 \left(1 - \frac{1}{n} \right) - \frac{q}{n}.
\end{equation*}
Furthermore, this condition will be always satisfied provided that $n \geq 2$ 
because it holds even when $\DegreeSecondMoment$ takes its minimum value, namely \cite{Ferrer2013b}
\begin{equation*}
\DegreeSecondMoment^{linear} = 4-\frac{6}{2}.
\end{equation*}
The substitution by $\DegreeSecondMoment^{linear}$ in the condition above gives $q \geq 2$, which is trivially true for any tree such that $n \geq 2$ as any tree with such a number of vertices has at least two leaves. 

\subsection{Bistar trees}

The following corollary presents formulae of $D_{min}^t$ for bistar trees and three instances: stars, quasistars and balanced star trees. 
\begin{cor}
In any bistar tree, where $k_1$ is the largest degree, 
\begin{eqnarray}
D_{min}^{bistar} & = & \DminBistarTreeFloor \label{D_min_bistar_tree_floor_equation} \\
                 & = & \DminBistarTreeMod, \label{D_min_bistar_tree_mod_equation} 
\end{eqnarray}
where 
\begin{equation*}
q' = {k_1 \bmod 2} + {(n - k_1) \bmod 2}.  
\end{equation*}
In addition,
\begin{eqnarray}
D_{min}^{b-bistar} & = & \DminBalancedBistarTreeFloor \label{D_min_balanced_bistar_tree_floor_equation} \\
                   & = & \DminBalancedBistarTreeMod, \label{D_min_balanced_bistar_tree_mod_equation}
\end{eqnarray} 
with
\begin{equation*}
\phi = \DminBalancedBistarTreeModDecrementCarlos,
\end{equation*}
and also
\begin{eqnarray}
D_{min}^{quasi} & = & \DminQuasistarTreeFloor \nonumber \\
                & = & \DminQuasistarTreeMod, \label{D_min_quasi_tree_equation}
\end{eqnarray}
\begin{eqnarray}
D_{min}^{star} & = & \DminStarTreeFloor \nonumber \\
               & = & \DminStarTreeMod. \label{D_min_star_tree_equation}
\end{eqnarray}
\end{cor}
\begin{proof}
As a bistar tree is a caterpillar tree, the application of equation \ref{D_min_caterpillar_floor_equation} (Theorem \ref{D_min_caterpillar_theorem}) with $k_2 = n - k_1$ and $k_i = 1$ for $i \geq 3$,  gives equation \ref{D_min_bistar_tree_floor_equation}. Besides, the application of \ref{D_min_caterpillar_degree_2nd_moment_equation} (Theorem \ref{D_min_caterpillar_theorem}) with $\DegreeSecondMoment$ for a bistar tree (equation \ref{degree_2nd_moment_bistar_equation}) produces equation \ref{D_min_bistar_tree_mod_equation} after some mechanical work.

As a balanced bistar tree is a bistar tree with $k_1 = \left\lceil n/2 \right\rceil$, equation \ref{D_min_bistar_tree_floor_equation} gives 
\begin{equation*}
\left\lfloor \frac{1}{4}\left(\LargestDegreeBalancedBistarTree + 1 \right)^2 \right\rfloor + \left\lfloor \frac{1}{4}\left(\SecondLargestDegreeBalancedBistarTree + 1\right)^2 \right\rfloor - 1
\end{equation*}
immediately. However, a much simpler modular formula will be obtained from equation \ref{D_min_bistar_tree_mod_equation}. In this respect, 
notice that 
\begin{equation*}
k_1 = \frac{1}{2}(n + {n \bmod 2}),
\end{equation*}
and also that $0 \leq q' \leq 2$, in particular, $q' = 1$ if $n$ is odd and
\begin{equation*}
q' = 2 \left( {\frac{n}{2} \bmod 2} \right)
\end{equation*} 
if $n$ is even. Then equation \ref{D_min_bistar_tree_mod_equation} produces
\begin{equation*}
D_{min}^{b-bistar}
        =    \left\{
                \begin{array}{lr} 
                   \frac{1}{8}(n^2+4n- 8 + 4 (n/2 \bmod 2)) & n\mbox{~is even} \\ 
                   \frac{1}{8}(n^2+4n-5) & n\mbox{~is odd}.
                \end{array} 
             \right.
\end{equation*}
or 
\begin{equation*}
D_{min}^{b-bistar}
        =    \left\{
                \begin{array}{ll} 
                   \frac{1}{8}(n^2+4n-8) & n \bmod 4 = 0 \\ 
                   \frac{1}{8}(n^2+4n-4) & n \bmod 4 = 2 \\ 
                   \frac{1}{8}(n^2+4n-5) & \mbox{otherwise}.
                \end{array} 
             \right.
\end{equation*}
in expanded form.
From this point, equation \ref{D_min_balanced_bistar_tree_mod_equation} follows immediately. Noting that 
\begin{eqnarray*}
D_{min}^{b-bistar} = \frac{1}{8}[(n + 2)^2 - 8 - (n + 2)^2 \bmod 8]
\end{eqnarray*}
and applying the definition of modulus, i.e. 
\begin{eqnarray*}
(n + 2)^2 \bmod 8 = (n + 2)^2 - 8 \left\lfloor \frac{(n+2)^2}{8} \right\rfloor,
\end{eqnarray*}
one finally obtains \ref{D_min_balanced_bistar_tree_floor_equation}.
 
A quasistar tree is a bistar tree where $k_1 = n - 2$, which transforms equations \ref{D_min_bistar_tree_floor_equation} and \ref{D_min_bistar_tree_mod_equation} into equation \ref{D_min_quasi_tree_equation} after some algebraic work. 
Similarly, a star tree is a bistar tree where $k_1 = n - 1$, which transforms equations \ref{D_min_bistar_tree_floor_equation} and \ref{D_min_bistar_tree_mod_equation} into equation \ref{D_min_star_tree_equation}. 
Equation \ref{D_min_star_tree_equation} has been  derived through other means \cite{Esteban2016a}.
\end{proof}

\begin{figure}
\centering
\includegraphics[scale = 0.6]{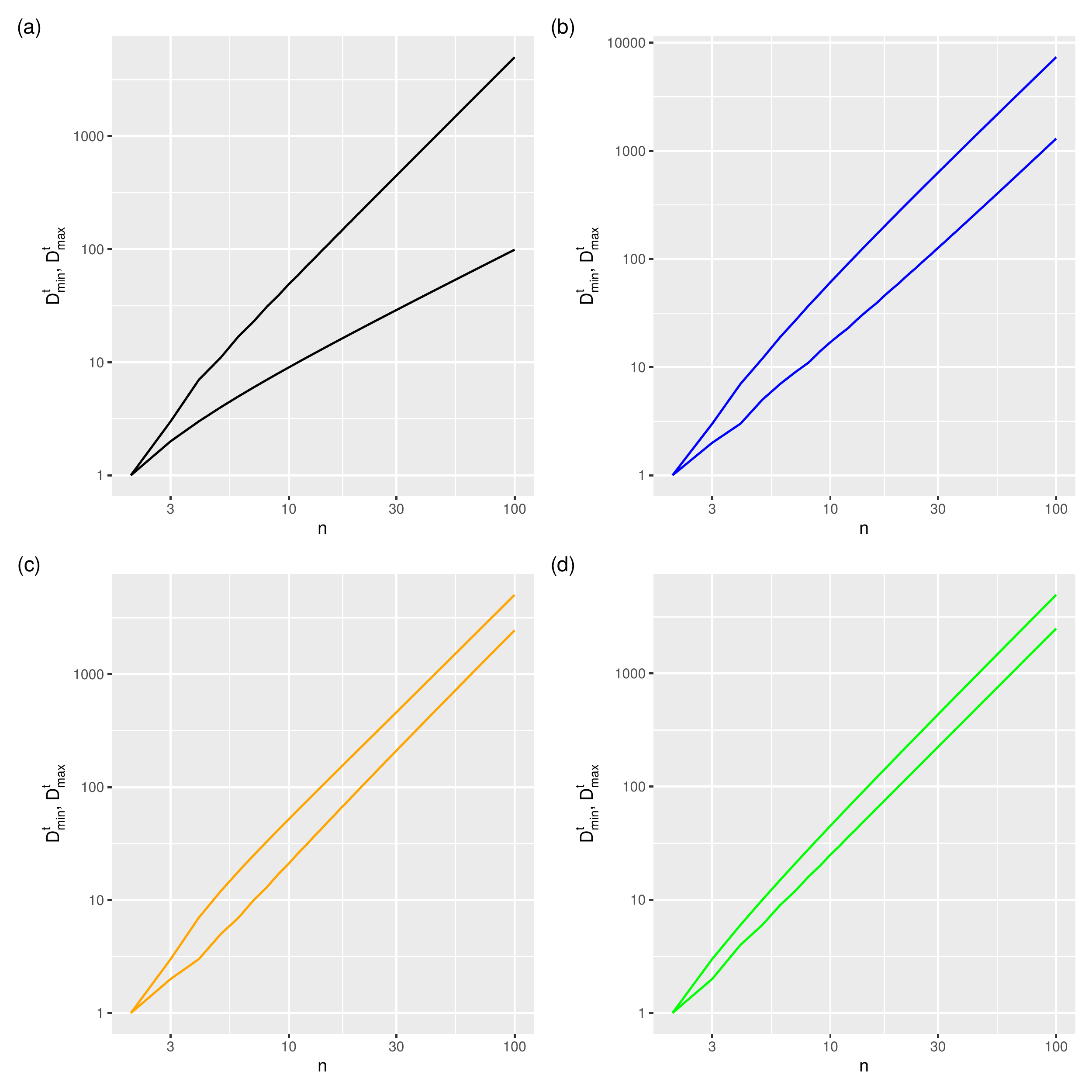}
\caption{\label{scaling_of_D_min_and_D_max_by_tree_figure} The scaling of $D_{min}^{t}$ and $D_{max}^{t}$ as a function of $n$, the number of vertices of the tree $t$, for different trees sorted increasingly by their hubiness (Table \ref{hubiness_table}). (a) linear trees (black). (b) balanced bistar trees (blue). (c) quasistar trees (orange) and (d) star trees (green). }
\end{figure}

Figure \ref{scaling_of_D_min_and_D_max_figure} (a) shows the growth of $D_{min}^t$ for different trees. Figure \ref{scaling_of_D_min_and_D_max_by_tree_figure} compares the growth of $D_{max}^t$ against $D_{min}^t$ for each tree.

%% file: maxima_of_optimality_scores.tex
\label{maxima_of_optimality_scores_section}

Here we aim to investigate a couple of optimality scores: $\Delta^t = D^t - D_{min}^t$ \cite{Gulordava2015} and $\Gamma^t=D^t/D_{min}^t$ \cite{Ferrer2004b,Tily2010a}. By definition of $D_{min}^t$,
$\Delta^t \geq \Delta_{min}^t = 0$ and $\Gamma^t \geq \Gamma_{min}^t = 1$. 

For a specific tree $t$ of $n$ vertices, the maximum value of $\Delta^t$ over all possible linear arrangements is 
\begin{equation*}
\Delta_{max}^t = D_{max}^{t} - D_{min}^{t}.
\end{equation*}
Similarly, 
\begin{equation*}
\Gamma_{max}^t = \frac{D_{max}^{t}}{D_{min}^{t}}.
\end{equation*}
Indeed, $\Delta_{max}^t$ is the vertical distance separating $D_{min}^t$ and $D_{max}^t$ for a given $n$ in Figure \ref{scaling_of_D_min_and_D_max_by_tree_figure}.

Table \ref{summary_table} allows one to obtain formulae of $\Delta_{max}^t$ or $\Gamma_{max}^t$ for specific trees.  
Figure \ref{scaling_of_delta_and_gamma_figure} shows the growth of $\Delta_{max}^t$ and $\Gamma_{max}^t$ for specific trees. The star tree is actually a baseline because we will show that it minimizes $\Delta_{max}^t$ and $\Gamma_{max}^t$. 
In star trees, quasistar trees and balanced bistar trees, $\Gamma_{max}^t$ converges to a constant (figure \ref{scaling_of_delta_and_gamma_figure} (b)) because both $D_{min}^t$ and $D_{max}^t$ are quadratic functions of $n$ (Table \ref{summary_table}). In balanced bistar trees, the leading coefficients are $1/8$ and $3/4$, respectively, which gives
\begin{equation*}
\lim_{n \rightarrow \infty} \Gamma_{max}^{b-bistar} = 6.
\end{equation*}    
By similar arguments,
\begin{equation*}
\lim_{n \rightarrow \infty} \Gamma_{max}^{star} = \lim_{n \rightarrow \infty} \Gamma_{max}^{quasi} = 2.
\end{equation*}
These limiting values are consistent with figure \ref{scaling_of_delta_and_gamma_figure} (b).

Here we aim to apply the results in the preceding sections to investigate an important question for research on these scores as a function of $n$ \cite{Ferrer2004b,Gulordava2015}: what are the minimum and the maximum value that $\Delta_{max}^t$ or $\Gamma_{max}^t$ can attain over all trees of $n$ vertices?



\begin{figure}
\centering
\includegraphics[scale = 0.8]{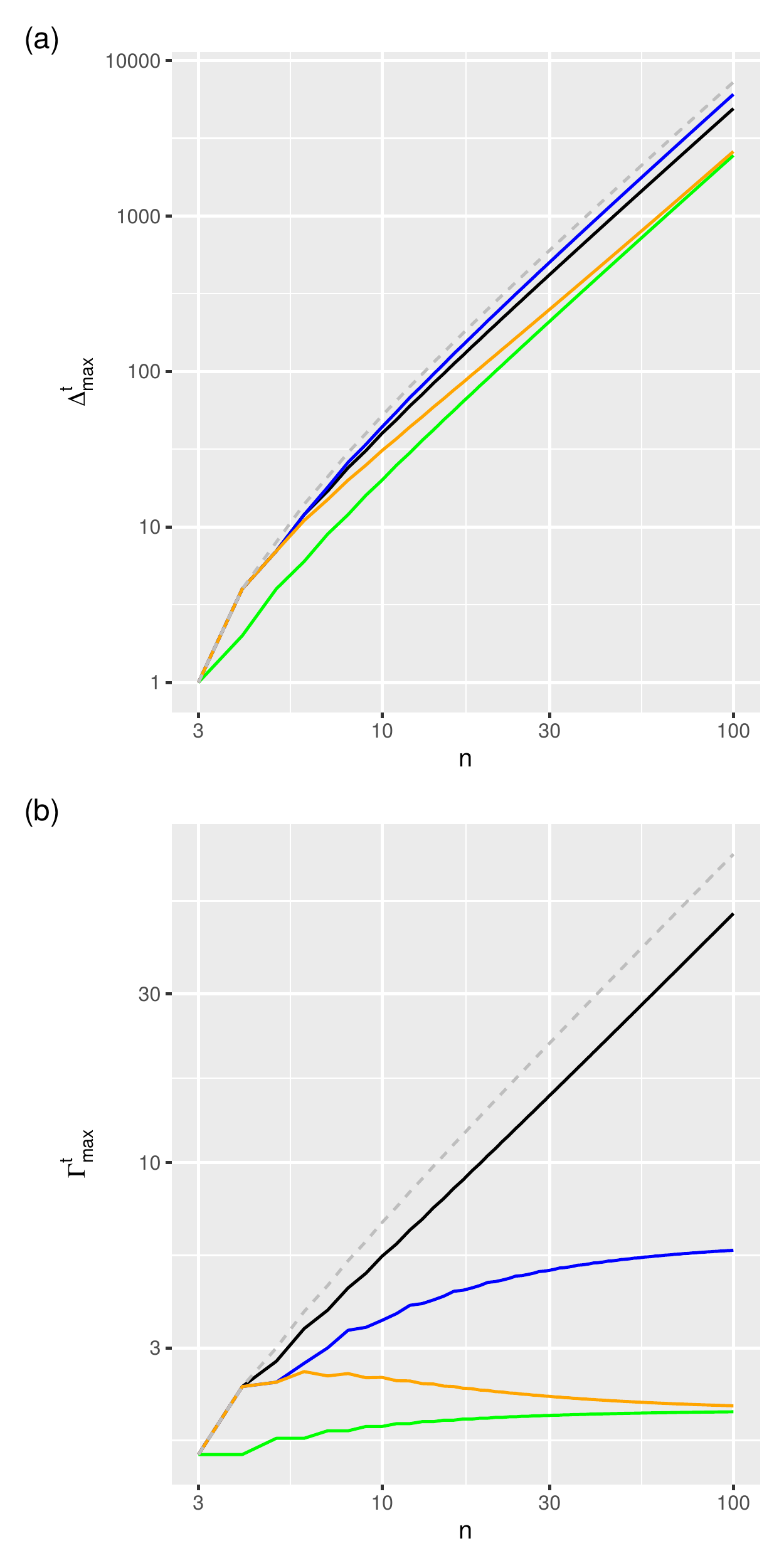}
\caption{\label{scaling_of_delta_and_gamma_figure} The scaling of the maximum value of the optimality scores as a function of $n$, the number of vertices of the tree $t$, for different trees: linear trees (black), balanced bistar trees (blue), quasistar trees (orange) and star trees (green). 
(a) $\Delta_{max}^{t}$. For reference, the upper bound $D_{max}^{b-bistar} - D_{min}^{linear}$ (dashed gray line) is also shown. 
(b) $\Delta_{max}^{t}$. For reference, the upper bound $D_{max}^{b-bistar}/D_{min}^{linear}$ (dashed gray line) is also shown.}
\end{figure}

\subsection{The minima of $\Delta_{max}^t$ and $\Gamma_{max}^t$}

The fact that $D_{max}^t \geq D_{max}^{star}$ (Theorem \ref{minimum_D_max_theorem}) and $D_{min}^t \leq D_{min}^{star}$ \cite{Iordanskii1974a,Esteban2016a} gives
\begin{eqnarray*}
D_{max}^t - D_{min}^{t} \geq D_{max}^{star} - D_{min}^{star} \\
\Delta_{max}^t \geq \Delta_{max}^{star} \\
\end{eqnarray*}
and also
\begin{eqnarray*}
\frac{D_{max}^t}{D_{min}^{t}} \geq \frac{D_{max}^{star}}{D_{min}^{star}} \\
\Gamma_{max}^t \geq \Gamma_{max}^{star}. \\
\end{eqnarray*}

\subsection{The maxima of $\Delta_{max}^t$ and $\Gamma_{max}^t$}

The fact that $D_{max}^t \leq D_{max}^{b-bistar}$ (Theorem \ref{maximum_D_max_theorem}) and $D_{min}^t \geq D_{min}^{linear}$ \cite{Esteban2016a} 
imply that
\begin{equation}
\Delta_{max}^t \leq D_{max}^{b-bistar} - D_{min}^{linear} \label{Delta_max_equation}
\end{equation}
and also
\begin{equation}
\Gamma_{max}^t \leq \frac{D_{max}^{b-bistar}}{D_{min}^{linear}}. \label{Gamma_max_equation}
\end{equation}
However, these are unlikely to be tight upper bounds of $\Delta_{max}^t$ and $\Gamma_{max}^t$ because the two kinds of trees involved in equation \ref{Delta_max_equation} and equation \ref{Gamma_max_equation}, star trees and balanced bistar trees, are not the same, contrary to what happened for the minima of $\Delta_{max}^t$ and $\Gamma_{max}^t$, given exactly by a star tree in both cases.

\begin{table}
\caption{\label{Delta_max_table} 
Maximum $\Delta_{max}^t$ as a function of $n$ and statistical properties of the trees that reach it: the kind of tree, $K_2 = n\left<k^2\right>$, the sum of squared degrees, $k_1$, the maximum degree, $n_1$, the number of leaves, $L^t$, the diameter in edges, $\left<l\right>^t$, the average path length, and finally, $D_{min}^t$ and $D_{max}^t$, the minimum and the maximum of $D^t$ over all $n!$ linear arrangements. As for the kind of tree, {\em quasi} stands for quasistar tree, {\em b-bistar} for balanced bistar and {\em cat} for caterpillar. 
}	
\begin{indented}
\item[]
\begin{tabular}{lrcrrrrrrrr}
\br
$n$ & $\Delta_{max}^{b-bistar}$ & Maximum & Kind of tree & $K_2$ & $k_1$ & $n_1$ & $L^t$ & $\left<l\right>^t$ & $D_{min}^t$ & $D_{max}^t$ \\
    &                           & $\Delta_{max}^t$ \\
\mr
\input{figures/maxima_of_Delta_max.tex}  
\br
\end{tabular}
\end{indented}
\end{table}

We perform a computational analysis of the maxima of $\Delta_{max}^t$ and $\Gamma_{max}^t$ (the methods are explained in \ref{validation_appendix}).
One the one hand, such analysis indicates that (Table \ref{Delta_max_table})
\begin{equation}
\Delta_{max}^t \leq \Delta_{max}^{b-bistar}
\label{Delta_max_conjecture1_equation}
\end{equation}
for $n \leq 8$, consistently with figure \ref{scaling_of_delta_and_gamma_figure} (a), but 
\begin{equation}
\Delta_{max}^t \leq \Delta_{max}^{t^*}
\label{Delta_max_conjecture2_equation}
\end{equation}
for $9 \leq n \leq 11$,
where $t^*$ is some caterpillar tree that is neither a bistar nor a linear tree. In addition, the bistar tree is not the only tree maximizing $\Delta_{max}^t$ for $4 \leq n \leq 8$ (Table \ref{Delta_max_table}). Notice that, for $n = 3$, a linear tree, a star tree and a balanced bistar tree are actually the same tree (when $n = 4$, the linear tree and the balanced bistar tree are the same tree). On the other hand, the analysis indicates that (Table \ref{Gamma_max_table}) 
\begin{equation}
\Gamma_{max}^t \leq \Gamma_{max}^{linear}
\label{Gamma_max_conjecture_equation}
\end{equation}
for $n \leq 11$, consistently with figure \ref{scaling_of_delta_and_gamma_figure} (b). Interestingly, the linear tree is the only tree maximizing $\Gamma_{max}^t$ up to $n = 11$ (Table \ref{Gamma_max_table}).

\begin{table}
\caption{\label{Gamma_max_table} 
Maximum $\Gamma_{max}^t$ as a function of $n$ and statistical properties of the trees that reach it. The format is based on that of Table \ref{Delta_max_table}. $\left<l\right>^t = (n+1)/3$ as expected for a linear tree \cite{Ferrer2003a}. }	
\begin{indented}
\item[]
\begin{tabular}{lrcrrrrrrrr}
\br
$n$ & $\Gamma_{max}^{linear}$ & Maximum & Kind of tree & $K_2$ & $k_1$ & $n_1$ & $L^t$ & $\left<l\right>^t$ & $D_{min}^t$ & $D_{max}^t$ \\
    &                           & $\Gamma_{max}^t$ \\
\mr
\input{figures/maxima_of_Gamma_max.tex}
\br
\end{tabular}
\end{indented}
\end{table}

\subsection{The relationship with $\Delta_{rla}$ and $\Gamma_{rla}$}

We define the expected value of $\Delta^t$ and $\Gamma^t$ in a random linear arrangement (rla) of a given tree $t$ as 
$\Delta_{rla}^t$ and $\Gamma_{rla}^t$ respectively.
Recall that $D_{rla} = \E_{rla}[D]$. Given a tree $t$, $D_{min}^t$ and $D_{rla}$ are constant and then
\begin{eqnarray*}
\Delta_{rla}^t & = & \E_{rla}[\Delta^t] \\
               & = & \E_{rla}[D^t - D_{min}^t] \\
               & = & D_{rla} - D_{min}^t \\   
\Gamma_{rla}^t & = & \E_{rla}[\Gamma^t] \\
               & = & \E_{rla}\left[\frac{D^t}{D_{min}^t}\right] \\
               & = & \frac{D_{rla}}{D_{min}^t}.
\end{eqnarray*}
The fact that 
\begin{equation*}
D_{min}^t \leq D_{rla} \leq D_{max}^t
\end{equation*}
gives
\begin{eqnarray*}
0 = \Delta_{min}^t \leq \Delta_{rla}^t \leq \Delta_{max}^t \\
1 = \Gamma_{min}^t \leq \Gamma_{rla}^t \leq \Gamma_{max}^t.
\end{eqnarray*}

%% file: figures/maxima_of_Delta_max.tex
3 & 1 & 1 &  linear star b-bistar & 6 & 2 & 2 & 2 & 1.33 & 2 & 3 \\ 
4 & 4 & 4 &  linear quasi b-bistar & 10 & 2 & 2 & 3 & 1.67 & 3 & 7 \\ 
5 & 7 & 7 &  linear & 14 & 2 & 2 & 4 & 2 & 4 & 11 \\ 
 & & &  quasi b-bistar & 16 & 3 & 3 & 3 & 1.8 & 5 & 12 \\ 
6 & 12 & 12 &  linear & 18 & 2 & 2 & 5 & 2.33 & 5 & 17 \\ 
 & & &  cat & 20 & 3 & 3 & 4 & 2.07 & 6 & 18 \\ 
 & & &  b-bistar & 22 & 3 & 4 & 3 & 1.93 & 7 & 19 \\ 
7 & 18 & 18 &  cat & 24 & 3 & 3 & 5 & 2.48 & 7 & 25 \\ 
 & & &  cat & 24 & 3 & 3 & 5 & 2.38 & 7 & 25 \\ 
 & & &  cat & 26 & 3 & 4 & 4 & 2.19 & 8 & 26 \\ 
 & & &  cat & 28 & 4 & 4 & 4 & 2.1 & 8 & 26 \\ 
 & & &  b-bistar & 30 & 4 & 5 & 3 & 2 & 9 & 27 \\ 
8 & 26 & 26 &  cat & 30 & 3 & 4 & 5 & 2.43 & 9 & 35 \\ 
 & & &  cat & 30 & 3 & 4 & 5 & 2.64 & 9 & 35 \\ 
 & & &  cat & 30 & 3 & 4 & 5 & 2.5 & 9 & 35 \\ 
 & & &  cat & 34 & 4 & 5 & 4 & 2.21 & 10 & 36 \\ 
 & & &  b-bistar & 38 & 4 & 6 & 3 & 2.07 & 11 & 37 \\ 
9 & 34 & 35 &  cat & 38 & 4 & 5 & 5 & 2.5 & 11 & 46 \\ 
 & & &  cat & 38 & 4 & 5 & 5 & 2.56 & 11 & 46 \\ 
 & & &  cat & 38 & 4 & 5 & 5 & 2.44 & 11 & 46 \\ 
 & & &  cat & 38 & 4 & 5 & 5 & 2.72 & 11 & 46 \\ 
 & & &  cat & 42 & 4 & 6 & 4 & 2.28 & 12 & 47 \\ 
10 & 44 & 46 &  cat & 46 & 4 & 6 & 5 & 2.47 & 13 & 59 \\ 
 & & &  cat & 46 & 4 & 6 & 5 & 2.82 & 13 & 59 \\ 
 & & &  cat & 46 & 4 & 6 & 5 & 2.56 & 13 & 59 \\ 
11 & 55 & 57 &  cat & 50 & 4 & 6 & 6 & 2.73 & 14 & 71 \\ 
 & & &  cat & 50 & 4 & 6 & 6 & 2.87 & 14 & 71 \\ 
 & & &  cat & 50 & 4 & 6 & 6 & 3.02 & 14 & 71 \\ 
 & & &  cat & 52 & 4 & 7 & 5 & 2.58 & 15 & 72 \\ 
 & & &  cat & 52 & 4 & 7 & 5 & 2.73 & 15 & 72 \\ 
 & & &  cat & 52 & 4 & 7 & 5 & 2.84 & 15 & 72 \\ 
 & & &  cat & 52 & 4 & 7 & 5 & 2.62 & 15 & 72 \\ 
 & & &  cat & 56 & 5 & 7 & 5 & 2.55 & 16 & 73 \\ 
 & & &  cat & 56 & 5 & 7 & 5 & 2.58 & 16 & 73 \\ 
 & & &  cat & 56 & 5 & 7 & 5 & 2.87 & 16 & 73 \\ 
 & & &  cat & 56 & 5 & 7 & 5 & 2.47 & 16 & 73 \\ 

%% file: figures/maxima_of_Gamma_max.tex
3 & 1.5 & 1.5 &  linear star b-bistar & 6 & 2 & 2 & 2 & 1.33 & 2 & 3 \\ 
4 & 2.33 & 2.33 &  linear quasi b-bistar & 10 & 2 & 2 & 3 & 1.67 & 3 & 7 \\ 
5 & 2.75 & 2.75 &  linear & 14 & 2 & 2 & 4 & 2 & 4 & 11 \\ 
6 & 3.4 & 3.4 &  linear & 18 & 2 & 2 & 5 & 2.33 & 5 & 17 \\ 
7 & 3.83 & 3.83 &  linear & 22 & 2 & 2 & 6 & 2.67 & 6 & 23 \\ 
8 & 4.43 & 4.43 &  linear & 26 & 2 & 2 & 7 & 3 & 7 & 31 \\ 
9 & 4.88 & 4.88 &  linear & 30 & 2 & 2 & 8 & 3.33 & 8 & 39 \\ 
10 & 5.44 & 5.44 &  linear & 34 & 2 & 2 & 9 & 3.67 & 9 & 49 \\ 
11 & 5.9 & 5.9 &  linear & 38 & 2 & 2 & 10 & 4 & 10 & 59 \\ 

%% file: zeta_score_section.tex
\label{zeta_score_section}

\begin{table}
\caption{\label{variance_table} 
$\V_{rla}^t$, the variance of $D^t$ in uniformly random linear arrangements of a tree $t$ of $n$ vertices for specific trees. $\V_{rla}^{linear}$ and $\V_{rla}^{star}$ are borrowed from \cite{Ferrer2018a}. $\V_{rla}^{b-bistar}$ and $\V_{rla}^{quasi}$ are derived from \ref{variance_of_sum_of_edge_lengths_tree_equation} and the corresponding value of $\DegreeSecondMoment$ in table \ref{hubiness_table}. }

\begin{indented}
\item[]
\begin{tabular}{lrr}
\br
$t$       & $\V_{rla}^t$ \\
\mr
linear    & $\frac{1}{90}(n - 2)(n + 1)(4n - 7)$ \\
 & & \\
balanced bistar & $\frac{1}{180}(n+1)\left[2(n-4)\lceil n/2 \rceil(\lceil n/2 \rceil - n) + n(n^2-n-14) + 12 \right]$ \\
 & & \\
quasistar & $\frac{1}{180}(n+1)[n((n-3)n + 10)-20]$ \\
 & & \\
star      & $\frac{1}{180}(n+1)(n-1)(n+2)(n-2)$ \\
 & & \\
\br
\end{tabular}
\end{indented}
\end{table}

For a specific tree $t$ of $n$ vertices, the minimum and the maximum values of $D_z^t$ over all possible linear arrangements are 
\begin{eqnarray*}
D_{z,min}^t = \frac{D_{min}^{t} - D_{rla}}{(\V_{rla}^t)^{1/2}} \\
D_{z,max}^t = \frac{D_{max}^{t} - D_{rla}}{(\V_{rla}^t)^{1/2}}.
\end{eqnarray*}
Table \ref{summary_table} and \cite{Ferrer2018a}
\begin{equation}
\V_{rla}^t = \frac{n+1}{45} \left[ (n-1)^2 + \left(\frac{n}{4} - 1 \right)n \left<k^2 \right> \right]
\label{variance_of_sum_of_edge_lengths_tree_equation}
\end{equation}
allow one to obtain formulae of $D_{z,min}^t$ and $D_{z,max}^t$ for specific trees. Table \ref{variance_table} summarizes $\V_{rla}^t$ in these trees.
Let us consider $D_{z,min}^{linear}$ as an example. The numerator of $D_{z,min}^{linear}$ is (Table \ref{summary_table} and equation \ref{sum_of_dependecy_lengths_random_equation})
\begin{equation*}
D_{min}^{linear} - D_{rla} = -\frac{1}{3}(n - 1)(n - 2) 
\end{equation*}
whereas the denominator is $\V_{rla}^{linear}$ (Table \ref{variance_table}). 
Then 
\begin{equation*}
D_{z,min}^{linear} = -(n - 1)\left[\frac{10(n - 2)}{(n + 1)(4n - 7)}\right]^{-1/2}.
\end{equation*}
Figure \ref{scaling_of_zeta_score_figure} shows the evolution of $D_{z,min}^t$ and $D_{z,max}^t$ as $n$ increases for specific trees.

In star trees, quasistar trees and balanced bistar trees, both $D_{min}^t-D_{rla}$, $D_{max}^t-D_{rla}$ and $(\V_{rla}^t)^{1/2}$ are quadratic functions of $n$ (Tables \ref{summary_table} and \ref{variance_table}). In balanced bistar trees, the leading coefficient of $D_{min}^t-D_{rla}$ is $1/8-1/3=-5/24$ and that of $(\V_{rla}^t)^{1/2}$ is $1/(6\sqrt{10})$, giving
\begin{equation*}
\lim_{n \rightarrow \infty} D_{z,min}^{b-bistar} = -\frac{5}{4} \sqrt{10}.
\end{equation*}    
In stars and quasistars, the leading coefficients are $1/4 - 1/3 = -1/12$ and $1/(6\sqrt{5})$. Hence
\begin{equation*}
\lim_{n \rightarrow \infty} D_{z,min}^{star} = \lim_{n \rightarrow \infty} D_{z,min}^{quasi} = -\frac{\sqrt{5}}{2}.
\end{equation*}
In balanced bistar trees, the leading coefficient of $D_{max}^t-D_{rla}$ is $3/4-1/3=5/12$ while that of $(\V_{rla}^t)^{1/2}$ is $1/(6\sqrt{10})$, giving 
\begin{equation*}
\lim_{n \rightarrow \infty} D_{z,max}^{b-bistar} = \frac{5}{2}\sqrt{10}.
\end{equation*}    
In stars and quasistars, the leading coefficients are $1/2 - 1/3 = 1/6$ and $1/(6\sqrt{5})$. Hence
\begin{equation*}
\lim_{n \rightarrow \infty} D_{z,max}^{star} = \lim_{n \rightarrow \infty} D_{z,max}^{quasi} = \sqrt{5}.
\end{equation*}
These limiting values are consistent with figure \ref{scaling_of_zeta_score_figure}.



\begin{figure}
\centering
\includegraphics[scale = 0.8]{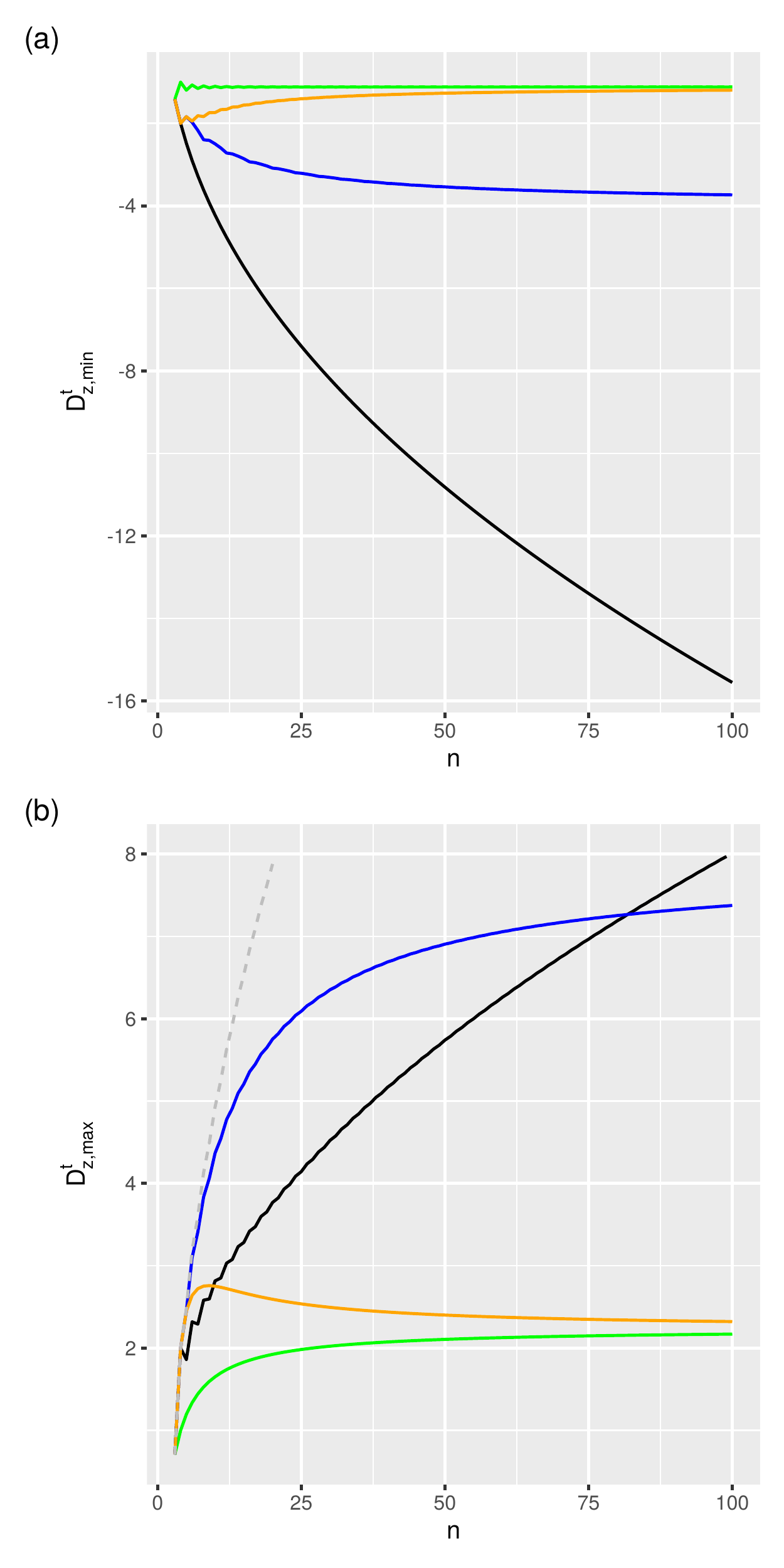}
\caption{\label{scaling_of_zeta_score_figure} The scaling of $D_{z,min}^{t}$ and $D_{z,max}^{t}$ as a function of $n$, the number of vertices of the tree $t$, for different trees: linear trees (black), balanced bistar trees (blue), quasistar trees (orange) and star trees (green). For reference, the upper bound of $D_{z,max}^{t}$, i.e. $(D_{max}^{b-bistar} - D_{rla})/(\V_{rla}^{linear})^{1/2}$ (dashed gray line) is also shown. }
\end{figure}

\subsection{The minima and the maxima of $D_{z,min}^t$}

Equation \ref{Dmin_chain_equation} in combination with \cite{Ferrer2018a}
\begin{equation}
\V_{rla}^{linear} \leq \V_{rla}^t \leq \V_{rla}^{star}
\label{D_variance_equation}
\end{equation}
yield that  
\begin{equation}
D_{z,min}^{linear} \leq D_{z,min}^t \leq D_{z,min}^{star}.
\end{equation}

\subsection{The minima and the maxima of $D_{z,max}^t$}

$D_{max}^{star} \leq D_{max}^t \leq D_{max}^{b-bistar}$ (Theorems \ref{maximum_D_max_theorem} and \ref{minimum_D_max_theorem}) in combination with equation \ref{D_variance_equation} yield
\begin{equation}
D_{z,min}^{star} \leq D_{z,max}^t \leq \frac{D_{max}^{b-bistar} - D_{rla}}{\V_{rla}^{linear}}.
\end{equation}
Again, the latter upper bound is unlikely to be a tight upper bound of $D_{z,max}^t$ because the two kinds of trees involved (a balanced bistar and a linear tree) are not the same. 
Interestingly, figure \ref{scaling_of_zeta_score_figure} (b) shows that $D_{z,max}^{bistar} \leq D_{z,max}^{linear}$ only for $n < 82$. 

\begin{figure}
\centering
\includegraphics[scale = 0.8]{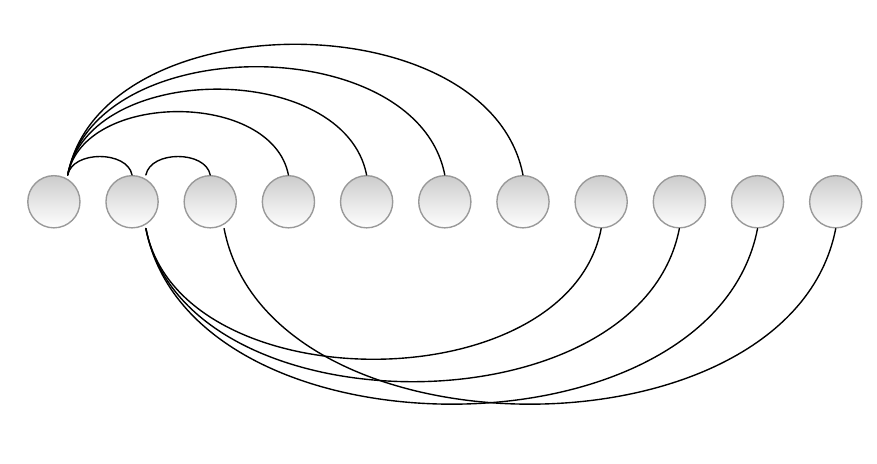}
\caption{\label{quasi_balanced_bistar_tree_figure} A linear arrangement of a quasibistar tree that maximizes $D_{z,max}^{t}$ for any tree of 11 vertices. }
\end{figure}

The computational analysis (\ref{validation_appendix}) in Table \ref{zeta_score_max_table} indicates that  
\begin{equation}
D_{z,max}^t \leq D_{z,max}^{b-bistar}
\label{zeta_score_max_conjecture_equation}
\end{equation}
for $n \leq 10$, consistently with figure \ref{scaling_of_zeta_score_figure} (b). 
The balanced bistar tree is the only tree maximizing $D_{z,max}^t$ up to $n = 10$ (Table \ref{zeta_score_max_table}). Contrary to expectations, the trend is broken for $n = 11$ because 
\begin{equation*}
D_{z,max}^t \leq D_{z,max}^{t^*}
\end{equation*}
with $t^*$ being some caterpillar tree that is neither a bistar nor a linear tree. That tree has only three internal vertices and is indeed a quasibistar tree (Fig. \ref{quasi_balanced_bistar_tree_figure}): it is a balanced bistar tree of $10$ vertices where one of the leaves has been connected to an extra vertex (leading to the degree sequence $k_1 = k_2 = 5$, $k_3 = 2$ and $k_i = 1$ for $4 \leq i \leq 11$).

\begin{table}
\caption{\label{zeta_score_max_table} 
Maximum $D_{z,max}^t$ as a function of $n$ and statistical properties of the trees that reach it. The format is based on that of Table \ref{Delta_max_table}. }	
\begin{indented}
\item[]
\begin{tabular}{lrcrrrrrrrr}
\br
$n$ & $D_{z,max}^{b-bistar}$ & Maximum & Kind of tree & $K_2$ & $k_1$ & $n_1$ & $L^t$ & $\left<l\right>^t$ & $D_{min}^t$ & $D_{max}^t$ \\
    &                           & $D_{z,max}^t$ \\
\mr
\input{figures/maxima_of_zeta_score_max.tex}
\br
\end{tabular}
\end{indented}
\end{table}

\subsection{The relationship with $D_{z,rla}$}

We define $D_{z,rla}^t$ as the expected the value of $D_z^t$ in a random linear arrangement (rla) of a given tree $t$.
As $D_{min}^t$, $D_{rla}$ and $\V_{rla}^t$ are constant given a tree $t$, one has 
\begin{eqnarray*}
D_{z,rla}^t & = & \E_{rla}\left[\frac{D^{t} - D_{rla}}{(\V_{rla}^t)^{1/2}}\right] \\
            & = & \frac{\E_{rla}[D^{t}] - D_{rla}}{(\V_{rla}^t)^{1/2}} \\
            & = & 0. 
\end{eqnarray*}
The fact that
\begin{equation*}
D_{min}^t \leq D_{rla} \leq D_{max}^t
\end{equation*}
gives
\begin{eqnarray*}
D_{z,min}^t \leq D_{z,rla}^t \leq D_{z,max}^t.
\end{eqnarray*}

%% file: figures/maxima_of_zeta_score_max.tex
3 & 0.71 & 0.71 &  linear star b-bistar & 6 & 2 & 2 & 2 & 1.33 & 2 & 3 \\ 
4 & 2 & 2 &  linear quasi b-bistar & 10 & 2 & 2 & 3 & 1.67 & 3 & 7 \\ 
5 & 2.45 & 2.45 &  quasi b-bistar & 16 & 3 & 3 & 3 & 1.8 & 5 & 12 \\ 
6 & 3.1 & 3.1 &  b-bistar & 22 & 3 & 4 & 3 & 1.93 & 7 & 19 \\ 
7 & 3.41 & 3.41 &  b-bistar & 30 & 4 & 5 & 3 & 2 & 9 & 27 \\ 
8 & 3.84 & 3.84 &  b-bistar & 38 & 4 & 6 & 3 & 2.07 & 11 & 37 \\ 
9 & 4.06 & 4.06 &  b-bistar & 48 & 5 & 7 & 3 & 2.11 & 14 & 48 \\ 
10 & 4.37 & 4.37 &  b-bistar & 58 & 5 & 8 & 3 & 2.16 & 17 & 61 \\ 
11 & 4.54 & 4.56 &  cat & 62 & 5 & 8 & 4 & 2.33 & 18 & 74 \\ 

%% file: discussion.tex
\label{discussion_section}

The main results of the preceding sections have been validated using a computational procedure described in \ref{validation_appendix}.

We have investigated the limits of the variation of $D^t$, the sum of edge lengths of trees of a given size $n$ (Table \ref{summary_table}). As for $D_{min}^t$, we have contributed with new formulae for the class of caterpillar trees that depend only on $n$ and the vertex degrees, complementing the pioneering research in \cite{Horton1997a}. These formulae have allowed us to obtain formulae for the subclass of bistar trees that depend only on $n$ and $k_1$, the maximum degree, which in turn have allowed us to obtain new formulae that depend only on $n$ for specific trees: quasistar trees and balanced bistar trees. 
\cite{Horton1997a} obtained a lower bound for $D_{min}^t$ (Table \ref{summary_table}) that gives actually the exact value of $D_{min}^t$ when $t$ is a caterpillar. We have contributed with a much shorter proof of the argument and showing that the lower bound is actually a significant improvement with respect to previous attempts to provide a lower bound of $D_{min}^t$ based on vertex degrees \cite{Petit2003a,Ferrer2013b}. Therefore, although $D_{min}^t$ can be calculated in polynomial time employing existing algorithms \cite{Shiloach1979,Esteban2015a,Chung1984}, $D_{min}^t$ can be calculated in constant time for caterpillar given trees of size $n$, $\DegreeSecondMoment$ and $q$ (Table \ref{summary_table}).  

As for $D_{max}^t$, we have not found a simple enough formula for the class of caterpillar trees but we have obtained one for the subclass of bistar trees as function of $n$ and $k_1$. 
Thanks to this work we have obtained new formulae that depend only on $n$ for specific trees: quasistar trees and balanced bistar trees (Table \ref{summary_table}). The new formula of $D_{max}^t$ for linear trees has been obtained employing an independent analysis. A unified derivation of $D_{max}^t$ for linear trees and bistar trees, as well as a general but simple formula of $D_{max}^t$ for caterpillar trees, should be the subject of future research. Finally, we delimited the range of variation of $D_{max}^t$, obtaining the following chain of inequalities
\begin{equation}
D_{rla} \leq D_{max}^{star} \leq D_{max}^t \leq D_{max}^{b-bistar}. 
\label{chain_of_inequalities_D_equation}
\end{equation}
The importance of this chain is two-fold. First, it indicates that the problem of maximizing $D^g$ and that of minimizing $D^g$ are not symmetric, because the corresponding chain for the minimization problem does not involve balanced bistar trees (equation \ref{Dmin_chain_equation}).
Second, it links the problem of maximizing $D^t$ without constraints (i.e. $D_{max}^t$) with the problem of maximizing $D^t$ under the planarity constraint (i.e. $D_{max,P}^t$), since $D_{max,P}^{t} \leq D_{max,P}^{linear} = D_{max}^{star}$ \cite{Ferrer2013b}. The finding indicates that any tree has a linear arrangement reaching the maximum possible $D^t$ for any tree under the planarity constraint, namely $D_{max}^t \geq D_{max,P}^{linear} = D_{max}^{star}$ (\cite{Ferrer2013b} did not address the question of whether $D_{max,P}^{t} = D_{max,P}^{linear} = D_{max}^{star}$ for any other tree $t$). Real syntactic dependency trees are almost planar in the sense that edge crossings are scarce \cite{Ferrer2017a} and the origin of such a characteristic is being debated \cite{Gomez2019a}.


In this article, we have established some mathematical foundations for the analysis and development of optimality scores based on $D^t$ and explored some implications for the limits of the variation of two scores: $\Gamma^t$ and $\Delta^t$. We have obtained the following chains of inequalities:
\begin{eqnarray}
0 = \Delta_{min}^t \leq \Delta_{rla}^t \leq \Delta_{max}^{star} \leq \Delta_{max}^t \label{chain_of_inequalities_Delta_equation} \\ 
1 = \Gamma_{min}^t \leq \Gamma_{rla}^t \leq \Gamma_{max}^{star} \leq \Gamma_{max}^t. \label{chain_of_inequalities_Gamma_equation}
\end{eqnarray}
We conjecture that $\Gamma_{max}^t \leq \Gamma_{max}^{linear}$ and that the linear tree is the only maximum of $\Gamma_{max}^t$ (Table \ref{Gamma_max_table}). A linear tree is the tree that minimizes the denominator of $\Gamma_{max}^t$. The numerator is maximized by a balanced bistar tree but it is easy to show (just using the formulae in Table \ref{summary_table}) that $\Gamma_{max}^t \leq \Gamma_{max}^{linear}$ for any tree $t$ that is a bistar. 
Similarly, we have obtained the following chains of inequalities for the $z$-score:
\begin{eqnarray}
D_{z,min}^{linear} \leq D_{z,min}^t \leq D_{z,min}^{star} \leq D_{z,rla} = 0 \label{chain_of_inequalities_zeta_score_min_equation} \\
0 = D_{z,rla} \leq D_{z,max}^{star} \leq D_{z,max}^t. \label{chain_of_inequalities_zeta_score_max_equation}
\end{eqnarray}
	
The problem of the trees that maximize $\Gamma_{max}^t$, $\Delta_{max}^t$ and $D_{z,max}^t$ should receive further investigation in two directions: characterizing the trees that maximize these scores (proving or refuting the conjectures above) or, at least, expanding the range of $n$ for which the true optima are known. We hope that our findings stimulate further research on optimality scores in linear arrangements.

%% file: appendix.tex
\section{A derivation of $D_{max}^{linear}$}

\label{D_max_linear_trees_appendix}

\begin{figure}
\centering
\includegraphics{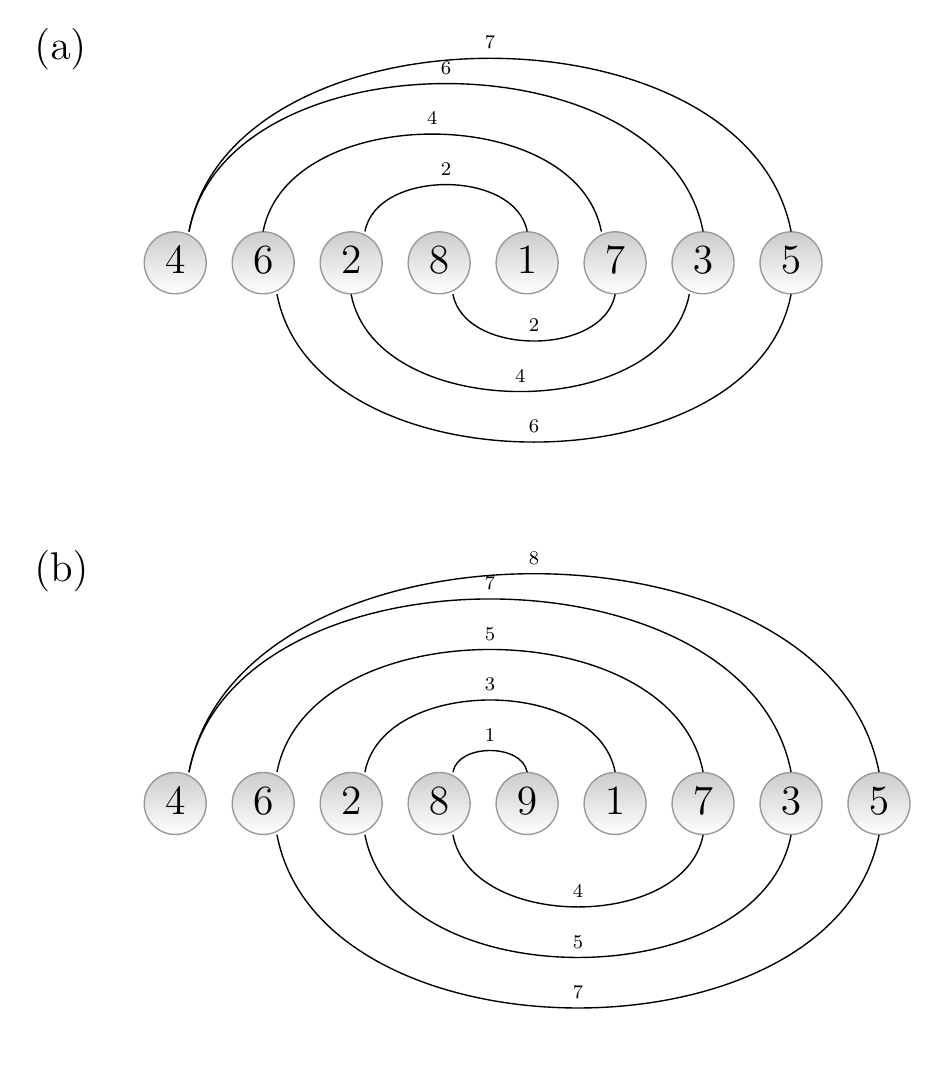} 
\caption{\label{linear_tree_figure} 
Maximum linear arrangements of linear trees with $n$ vertices. Vertex labels indicate the position of each vertex in the degree sequence. Edge labels indicate edge lengths. (a) $n = 8$ and $D^t = D_{max}^{linear} = 31$, generated by the permutation $\alpha = 5, 3, 7, 1, 8, 2, 6, 4$. (b) $n = 9$ and $D^t = D_{max}^{linear} = 39$, generated by the permutation $\alpha = 5, 6, 3, 8, 1, 9, 2, 7, 4$.}
\end{figure}

We apply a result by \cite{Chao_Liang1992a} to prove \ref{maximum_linear_arrangement_of_linear_tree_equation}.
Given a set $A_n=\{a_1,a_2,\dots,a_n\}$, where $a_1<a_2<\dots<a_n$, \cite{Chao_Liang1992a} shows how to calculate a permutation $\alpha=\alpha_1\alpha_2\dots\alpha_n$ such that for certain functions $f$,
\begin{equation*}
D_f(\alpha)=\sum_{i=1}^{n-1}f(|\alpha_i-\alpha_{i+1}|)
\end{equation*}
is maximized. $D_{max}^{linear}$ is a particular case with $A_n=\{1,2,\dots,n\}$ and $f$ the identity function ($id$). $A_n$ is the set of vertex labels when the vertices are labelled by a a depth first search traversal from one leaf to the other leaf assigning consecutive numbers between $1$ and $n$. 


Under these assumptions, for each permutation $\alpha$ of the form noted above, we can construct a linear arrangement whose cost is $D_{id}(\alpha)$. This is the arrangement where the $i$-th vertex (the vertex labelled with $i$) is assigned the position $\alpha_i$ in the linear arrangement. Thus, the length of each arc of the form $\{i,i+1\}$ is $|\alpha_{i} - \alpha_{i+1}|$, and the total sum of lengths is $\sum_{i=1}^{n-1}|\alpha_i-\alpha_{i+1}|$, which equals $D_{id}(\alpha)$. This correspondence between permutations and linear arrangements is trivially bijective, as one can go from linear arrangements to permutations following the inverse process.

We restate Theorem 1 in \cite{Chao_Liang1992a} under our particular conditions as follows:
\begin{thm}[Chao and Liang, \cite{Chao_Liang1992a}]
\label{teo_chao_liang}
A permutation of $\{1,2,...,n\}$ is maximum if it maximizes
\begin{equation*}
D(\alpha)=\sum_{i=1}^{n-1}|\alpha_i-\alpha_{i+1}|.
\end{equation*}
If $n=2c$, then the maximum permutations with $\alpha_1>\alpha_n$ are those satisfying the following three conditions:
\begin{enumerate}
\item[(i)] $\alpha_1=c+1$, $\alpha_n=c$
\item[(ii)] $\alpha_2\alpha_4\cdots\alpha_{2c-2}$ is a permutation of $\{1,2,\dots,c-1\}$
\item[(iii)] $\alpha_3\alpha_5\cdots\alpha_{2c-1}$ is a permutation of $\{c+2,c+3,\dots,2c\}$
\end{enumerate}
If $n=2c+1$, then the maximum permutations with $\alpha_1>\alpha_n$ are those satisfying the following three conditions:
\begin{enumerate}
\item[(iv)] $\alpha_1=c+1$, $\alpha_n=c$
\item[(v)] $\alpha_2\alpha_4\cdots\alpha_{2c}$ is a permutation of $\{c+2,c+3,\dots,2c+1\}$
\item[(vi)] $\alpha_3\alpha_5\cdots\alpha_{2c-1}$ is a permutation of $\{1,2,\dots,c-1\}$
\end{enumerate}
or the following three conditions:
\begin{enumerate}
\item[(vii)] $\alpha_1=c+2$, $\alpha_n=c+1$
\item[(viii)] $\alpha_2\alpha_4\cdots\alpha_{2c}$ is a permutation of $\{1,2,\dots,c\}$
\item[(ix)] $\alpha_3\alpha_5\cdots\alpha_{2c-1}$ is a permutation of $\{c+3,c+4,\dots,2c+1\}$
\end{enumerate}
The maximum permutations with $\alpha_1<\alpha_n$ are the reverse permutations of those specified above.
\end{thm}
Notice that according to conditions (i-iii) and (vii-ix) of the previous theorem, any linear tree has maximum linear arrangements that are divided into two parts: a first part with all the vertices with even labels and a second part with all the vertices with odd labels. Based on that property, we describe a procedure that generates concrete maximum linear arrangements such that they are easy to draw (figure \ref{linear_tree_figure}) and allow one to calculate $D_{max}^{linear}$ easily.

Let $c$ be $n/2$ and $\beta$ a function such that $\beta(x) = 1$ if $x < c$ and $\beta(x) = -1$ if $x > c$. When $n$ is even, the procedure is
\begin{enumerate}
\item 
Place the two leaves at the center of the linear arrangement (the vertex labelled with $n$ in position $c$ and the one labelled with $1$ in position $c+1$). This satisfies condition (i) of theorem \ref{teo_chao_liang}.
\item 
One leaf is the current vertex in the odd part and the other is the current vertex in the even part. 
\item
Repeat the following steps untill all vertices have been placed. 
  \begin{enumerate}
  \item
  Take the current vertex of the even part, say $x$, and place vertex $x+\beta(c)$ in the nearest free position in the odd part. This satisfies condition (iii) of theorem \ref{teo_chao_liang}. 
  \item
  Take the current vertex of the odd part, say $y$, and place vertex $y+\beta(y)$ in the nearest free position in the even part. This satisfies condition (ii) of theorem \ref{teo_chao_liang}. 
  \item 
  $x+\beta(x)$ becomes the current vertex of the odd part.
  \item
  $y+\beta(y)$ becomes the current vertex of the even part. 
  \end{enumerate}
\end{enumerate}
Figure \ref{linear_tree_figure} (a) shows the outcome of the procedure for $n = 8$, producing 2 arcs of length 2, 2 arcs of length 4, 2 arcs of length 6 and one arc of length 7. By adding each of the lengths produced for a linear tree with even $n$, one obtains
\begin{eqnarray}
D_{max}^{linear} & = & n-1 + 2\times(n-2) + 2\times(n-4) + \dots + 2\times 4 + 2\times 2 \nonumber \\ 
                 & = & n-1 + 2\times[ (n-2) + (n-4) + \dots + 4 + 2 ] \nonumber \\
                 & = & n-1 + 2\times\left[\frac{n-2}{2}\times\frac{2+n-2}{2}\right] \nonumber\\
	         & = & n-1 + \frac{n^2-2n}{2} \nonumber \\ 
	         & = & \frac{n^2-2}{2}. \label{linear_even_compact}
\end{eqnarray}
When $n$ is odd, the procedure is
\begin{enumerate}
\item 
Generate a linear arrangement for $n - 1$ vertices, i.e. $\alpha_1 \alpha_2...\alpha_{n-1}$ with the procedure above. That arrangement consists of an even part and odd part of $(n-1)/2$ vertices each. 
\item
Generate a linear arrangement for $n$ vertices from that of $n-1$ vertices by inserting vertex $n$ in the central position with respect to the linear arrangement of $n$ vertices, namely position $c + 1 = (n + 1)/2$ while displacing all vertices in the odd part one position to the right.    
\end{enumerate}
It is easy to see that the linear arrangement over $n$ vertices will be such that $\alpha_n = c + 1 = (n + 1)/2$ and $\alpha_1 = c + 2$ as expected from condition (viii) of theorem \ref{teo_chao_liang}
while there will be $(n-1)/2$ even vertices followed by $(n+1)/2$ odd vertices, thus satisfying conditions (viii) and (ix) of theorem \ref{teo_chao_liang}. 
Figure \ref{linear_tree_figure} (b) shows the outcome of the procedure for $n = 9$. When $n$ is odd (figure \ref{linear_tree_figure}) and reasoning analogously, we can build a maximum linear arrangement with a summation of lengths as the following
\begin{eqnarray}
D_{max}^{linear} & = & n-1 + 2\times(n-2) + 2\times(n-4) + \dots + 2\times 3 + 1 \nonumber \\ 
                 & = & n-1 + 2\times[ (n-2) + (n-4) + \dots + 5 + 3 ] + 1 \nonumber \\
                 & = & n + 2\times\left[\frac{n-3}{2}\times\frac{3+n-2}{2}\right] \nonumber \\
	         & = & n + \frac{n^2-2n-3}{2} \nonumber \\
	         & = & \frac{n^2-3}{2}. \label{linear_odd_compact}
\end{eqnarray} 
Finally, equations \ref{linear_even_compact} and ~\ref{linear_odd_compact} can be unified as equation \ref{maximum_linear_arrangement_of_linear_tree_equation}.

\begin{figure}
\centering
\includegraphics{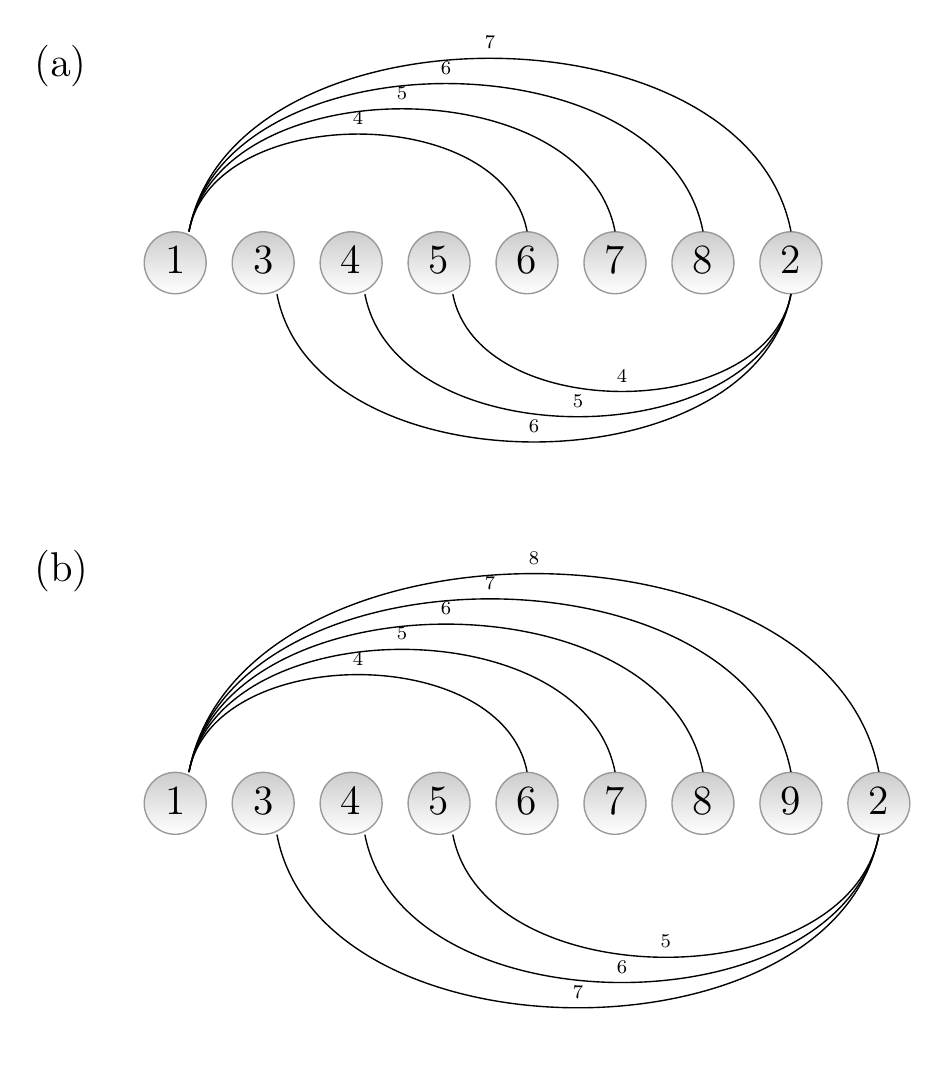} 
\caption{\label{bistar_figure} Extreme linear arrangements of balanced bistar trees of $n$ vertices. Vertex labels indicate the position of each vertex in the degree sequence. Edge labels indicate edge lengths. (a) $n = 8$ and $D^t = D_{max}^{b-bistar} = 37$. (b) $n = 9$ and $D^t = D_{max}^{b-bistar} = 48$. }
\end{figure}

\section{$D_{max}^{t}$ in bistar trees}

\label{D_max_bistar_tree_appendix}

We define $\gamma(i)$ as the set of adjacent vertices of $i$ \cite{Bollobas1998a}, also termed the set of 1st neighbours or nearest neighbours of $i$ \cite{Bogunya2003a}.
We define an extreme linear arrangement of a bistar tree as an ordering of the vertices following the one of the following templates: 
\begin{equation*}
1, \gamma(2)\setminus\{1\}, \gamma(1)\setminus\{2\},2  
\end{equation*}
as in figure \ref{bistar_figure}, or its symmetric, i.e.
\begin{equation*}
2, \gamma(1)\setminus\{2\}, \gamma(2)\setminus\{1\},1  
\end{equation*}

The following lemma indicates how to arrange a single vertex and its attached vertices so as to maximize its sum of edge lengths.
\begin{lem}
Suppose that $D_i^g$ is the sum of the lengths of the edges attached to the $i$-th vertex of a graph $g$ of $n$ vertices. 
Then $D_{i,max}^g$, the maximum value of $D_i^g$ over the $n!$ linear arrangements of the whole graph is
\begin{equation}
D_{i,max}^g = \frac{1}{2}k_i(2n - k_i - 1) 
\label{D_max_single_vertex_equation}
\end{equation}
\label{D_max_single_vertex_lemma}
and is achieved when vertex $i$ is placed at one of the ends of the linear arrangement and its adjacent vertices at the other end.
\end{lem}
\begin{proof}
When the $i$-th vertex is placed at one of the ends of the linear arrangement and 
its adjacent vertices as far as possible (consecutively at the other end), 
\begin{equation*}
D_i^g = \sum_{j=1}^{k_i} (n-j),
\end{equation*}
which gives equation \ref{D_max_single_vertex_equation}. If the $i$-th vertex is not placed at one of the ends but its neighbours are still placed as far as possible, 
$D_i^g$ cannot exceed $D_{i, max}^g$. A detailed argument follows.  

We define $h$ as the position of vertex $i$ in the linear arrangement ($1 \leq h \leq n$), $k_i^-$ as the number of neighbours of $i$ placed before $i$ and $k_i^+$ as the number of neighbours of $i$ placed after $i$. In such a linear arrangement, the maximum value of $D_i^g$, i.e. $D_{i,max, k_i^-, k_i^+}^g$, is achieved placing the $k_i^-$ neighbours at the beginning of the linear arrangement and the $k_i^+$ neighbours at the end of the linear arrangement, producing
\begin{equation*}
D_{i,max, k_i^-, k_i^+}^g = \sum_{j=1}^{k_i^-} (h - j) + \sum_{j=1}^{k_i^+} (n - j + 1 - h).
\end{equation*}
We will show that $D_{i,max, h, k_i^-, k_i^+}^g \leq D_{i,max}^g = D_{i,max, 1, 0, k_i}^g = D_{i,max, n, k_i, 0}^g$, i.e. 
\begin{equation*}
\sum_{j=1}^{k_i^-} (h - j) + \sum_{j=1}^{k_i^+} (n - j + 1 - h) \leq \sum_{j=1}^{k_i} (n-j),
\end{equation*}
that is equivalent to
\begin{equation*}
\sum_{j=1}^{k_i-k_i^+} (h - j) + (1 - h)k_i^+ \leq \sum_{j= k_i^+ + 1}^{k_i} (n-j).
\end{equation*}
Rearranging the terms and calculating certain summations the inequality becomes 
\begin{equation*}
0 \leq (k_i - k_i^+)(n - h ) + \sum_{j = 1}^{k_i - k_i^+} j - \sum_{j= k_i^+ + 1}^{k_i} j +  (h - 1)k_i^+.
\end{equation*}
Calculating the remaining summations one obtains, after some routine calculations, 
\begin{equation*}
0 \leq (k_i - k_i^+)(n - h - k_i^+) + (h - 1)k_i^+,
\end{equation*}
which is trivially true because $k_i^+ \leq k_i$, $1 \leq h$ and $k_i^+ \leq n - h$ by definition.
\end{proof}
The previous lemma generalizes a previous result on $D_{max}^{star}$, that is achieved when the hub of the star is located at one of the ends of the linear arrangement \cite{Ferrer2013e} (figure \ref{maximum_linear_arrangements_figure} (a)). In a star tree $t$, $D^t$ is determined by the sum of edge lengths of the hub vertex. 

The following lemma indicates that an extreme linear arrangement of a bistar is actually a maximum linear arrangement.  
\begin{lem}
In a bistar tree $t$ of $n$ vertices and maximum degree $k_1$, $D_{max}^t$ is  
\begin{equation}
D_{max}^{bistar} = k_1 (n - k_1) + \frac{n}{2}(n - 3) + 1
\label{D_max_bistar_tree_equation}
\end{equation}
\label{D_max_bistar_tree_theorem}
and a extreme linear arrangement of $t$ is actually a maximum linear arrangement.
\end{lem}
\begin{proof}
A bistar tree can be seen as two star trees joined by a common edge. Then $D^t$ can be decomposed as
\begin{equation}
D^t = D_1^t + D_2^t - d_{12}^t,
\label{D_decomposition_bistar_tree_equation}
\end{equation}
where $D_i^t$ is the sum of the lengths of edges attached to the vertex with the $i$-th largest degree and $d_{12}^t$ is the length of the edge joining the two vertices with the largest degrees.
To maximize $D^t$ following equation \ref{D_decomposition_bistar_tree_equation}, one has to maximize $D_1^t$ and $D_2^t$. By lemma \ref{D_max_single_vertex_lemma}, $D_1^t$ is maximized placing vertex 1 at one end and its neighbours at the other end. Since $D_2^t$ also must be maximized then, by the same lemma, vertex 2 has to be placed at the opposite end (otherwise $D_1^t < D_{1,max}^t$ or $D_2^t < D_{2,max}^t$), which gives $d_{12}^t = n - 1$. Such a linear arrangement is an extreme linear arrangement of a bistar tree and  
equation \ref{D_decomposition_bistar_tree_equation} gives
\begin{eqnarray*}
D_{max}^{bistar} & = & D_{1,max}^t + D_{2,max}^t - (n - 1) \\
                 & = & \frac{1}{2}k_1(2n - k_1 - 1) + \frac{1}{2}(n - k_1)(n + k_1 - 1) - (n - 1).
\end{eqnarray*}
Equation \ref{D_max_bistar_tree_equation} is recovered after some algebra.
\end{proof}
Thanks to the preceding work, formulae of $D_{max}^t$ for specific bistar trees follow easily.
\begin{cor}
\begin{eqnarray}
D_{max}^{b-bistar} = \DmaxBalancedBistarTree \label{extreme_sum_of_edge_lengths_equation} \\
D_{max}^{quasi} = \DmaxQuasistarTree \nonumber \\
D_{max}^{star} = \DmaxStarTree. \nonumber
\end{eqnarray}
\end{cor}
\begin{proof}
$D_{max}^{b-bistar}$ is obtained applying $k_1 = \left\lceil n/2 \right\rceil$ to equation \ref{D_max_bistar_tree_equation} (Theorem \ref{D_max_bistar_tree_theorem}).
When $n$ is odd, $k_1 = (n + 1)/2$ and then equation \ref{D_max_bistar_tree_equation} gives
\begin{equation*}
D_{max}^{b-bistar} = \frac{3}{4}(n-1)^2.
\end{equation*}
When $n$ is even, $k_1 = n/2$, one obtains
\begin{equation*}
D_{max}^{b-bistar} = \frac{1}{4}(3(n-1)^2 + 1).
\end{equation*}
Therefore, for any $n$, $D_{max}^{b-bistar}$ follows equation \ref{extreme_sum_of_edge_lengths_equation}.  
Similarly, $D_{max}^{quasi}$ is obtained with $k_1 = n - 2$ and $D_{max}^{star}$ is obtained with $k_1 = n - 1$ after some mechanical work.  
\end{proof}

\section{The minimum $D_{max}^t$}

\label{minimum_D_max_appendix}

\begin{figure}
\centering
\includegraphics[scale = 0.7]{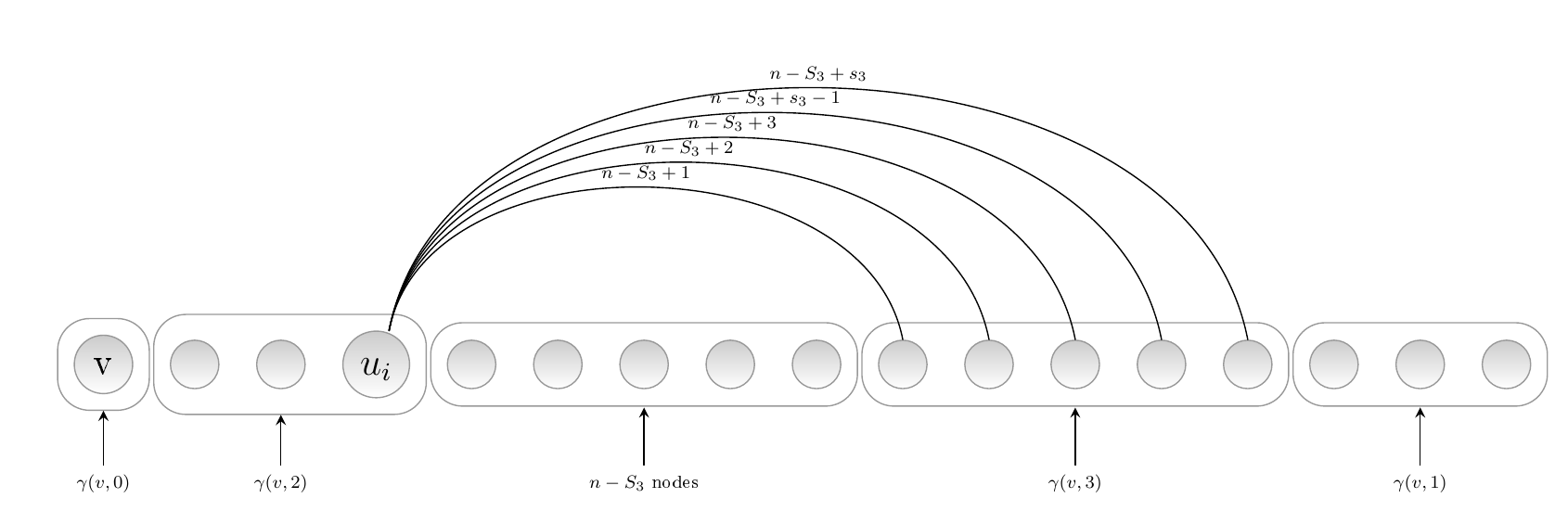} 
\caption{\label{induction_figure} 
The edges between $\gamma(v, 2)$ and $\gamma(v, 3)$ when they are formed exclusively with $u_i$, the vertex of $\gamma(v, 2)$ that is the closest to $\gamma(v, 3)$.
}
\end{figure}

The following theorem states that for any tree with $n$ vertices we can always find an arrangement with total length at least $\frac{n(n-1)}{2}$. 

\begin{thm}[Minimum $D_{max}^t$]
\label{minimum_D_max_theorem}
For any tree $t$ of $n$ vertices, 
\begin{equation*}
D_{max}^t \ge D_{max}^{star} = \DmaxStarTree.
\end{equation*}
\end{thm}

\begin{proof}

Let us consider a tree $t$ of $n$ vertices and define, $\gamma(v,i)$, the set of vertices that are at topological distance $i$ from vertex $v$ in $t$, with $\gamma(v,0) = \{ v \}$. 
Equivalently, $\gamma(v,i)$ is the set of $i$-th neighbours of $v$. $\gamma(v,1)$ is the set of vertices adjacent to $v$ \cite{Bollobas1998a}. 
For instance, in a star tree of $n$ vertices where the hub is vertex $1$ and $v \neq 1$ is some leaf, 
$\gamma(v,1) = \{1\}$ and $\gamma(v,2) = \{2,3,...,\}\setminus\{v\}$.

A linear arrangement that gives sum at least $D_{max}^{star}$ follows the template defined by the sequence
\begin{equation}
\gamma(v,0),\gamma(v,2),\gamma(v,4),\cdots,\gamma(v,3),\gamma(v,1)
\label{linear_arrangement_template_equation}
\end{equation}
This is not a proper arrangement because the $\gamma(v,i)$ is a set and its elements are not ordered. We can get a proper arrangement by ordering the vertices in every set $\gamma(v,i)$ in any arbitrary way.

Let $s_i$ be $|\gamma(v,i)|$ and $S_i=\sum_{j=0}^{i}s_j$.
We define $V_i$ as the set of vertices reached up to topological distance $i$, i.e.  
\begin{equation*}
V_i = \cup_{j=0}^i \gamma(v,j).
\end{equation*} 
Hence $S_i = |V_i|$.

Let us use induction on the topological distance $i$. 

{\em Induction hypothesis}. The sum of the lengths of the edges formed by vertices in $V_i$ is at least 
\begin{equation}
\sum_{j=1}^{S_i-1}(n - j).
\label{suma_parcial}
\end{equation}

{\em Base case.} For $i=0$, $S_i = 1$ and the sum of edge lengths must be zero trivially.

{\em Induction step.} Note that the number of vertices between $\gamma(v,i)$ and $\gamma(v,i+1)$, that is $s_{i+2}+s_{i+3}+\dots$, is $n-S_{i+1}$ (figure \ref{induction_figure}). 
According to the template of linear arrangement in equation \ref{linear_arrangement_template_equation}, the vertices in $\gamma(v,i+1)$ are the farthest away from those of $\gamma(v,i)$ among vertices with topological distance $i+1$ or more. 
Let $B$ the sum of lengths of the $s_{i+1}$ edges from $\gamma(v,i)$ to $\gamma(v,i+1)$. Suppose that these edges start from $u_i$, the vertex in $\gamma(v,i)$ nearer to $\gamma(v,i+1)$ in the linear arrangement, which implies the vertices in $\gamma(v,i)\setminus\{u_i\}$ must be leaves (figure \ref{induction_figure}). Then 
\begin{eqnarray*}
B & = & \sum_{j'=1}^{s_{i+1}} (n - S_{i+1} + j') \\
  & = & \sum_{j=S_{i+1} - s_{i+1}}^{S_{i+1}-1} (n - j) \\
  & = & \sum_{j=S_i}^{S_{i+1} - 1} (n - j).
\end{eqnarray*}
If edges from $\gamma(v,i)$ to $\gamma(v,i+1)$ involved any vertex in $\gamma(v,i)\setminus\{u_i\}$, then $B$ would increase.
Thus, thanks to the induction hypothesis, the sum of the costs of the edges from the vertices in $V_{i+1}$ is at least
\begin{eqnarray*}
\sum_{j=1}^{S_i-1}(n - j) + B & = & \sum_{j=1}^{S_i-1}(n - j) + \sum_{j=S_i}^{S_{i+1}-1} (n - j) \\ 
                                 & = & \sum_{j=1}^{S_{i+1}-1} (n - j) 
\end{eqnarray*}
as expected.

Let $\zeta(v)$ be the maximal topological distance to $v$ in some tree, then $S_{\zeta(u)} = n$ and equation \ref{suma_parcial} gives 
\begin{equation*}
D_{max}^t \geq \sum_{j=1}^{n - 1} j = \DmaxStarTree.
\end{equation*}
\end{proof}

\begin{figure}
\centering
\includegraphics{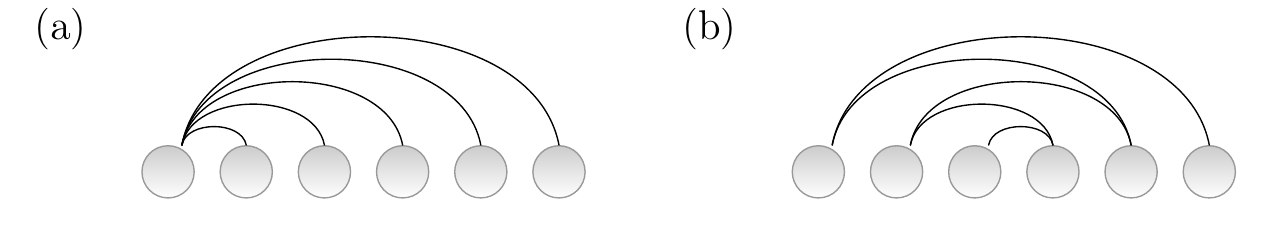} 
\caption{\label{maximum_linear_arrangements_figure} Linear arrangements of trees of $n$ vertices where $n = 6$ and $D^t = {n \choose 2} = 15$. (a) Linear arrangement of a star tree that is both a maximum linear arrangement and a maximum planar linear arrangement ($D^t = D_{max,P}^t = D_{max}^t$). (b) Linear arrangement of a linear tree that is a maximum planar linear arrangement but not a maximum linear arrangement ($D^t = D_{max,P}^t < D_{max}^t = 17$; recall Table \ref{summary_table}). }
\end{figure}

The previous theorem indicates that $D_{max}^t$ is at least its value for star trees (figure \ref{maximum_linear_arrangements_figure}). However, it is well-known that $D_{max}^{star}$ can also be achieved by a linear tree arranged as in figure \ref{maximum_linear_arrangements_figure} \cite{Ferrer2013b}. That arrangement follows from applying the template of arrangement in equation \ref{linear_arrangement_template_equation} with one of the leaves as the initial vertex.

\section{An alternative derivation of $D_{max}^t \leq D_{max}^{b-bistar}$}

\label{D_max_appendix}

\begin{figure}
\centering
\includegraphics{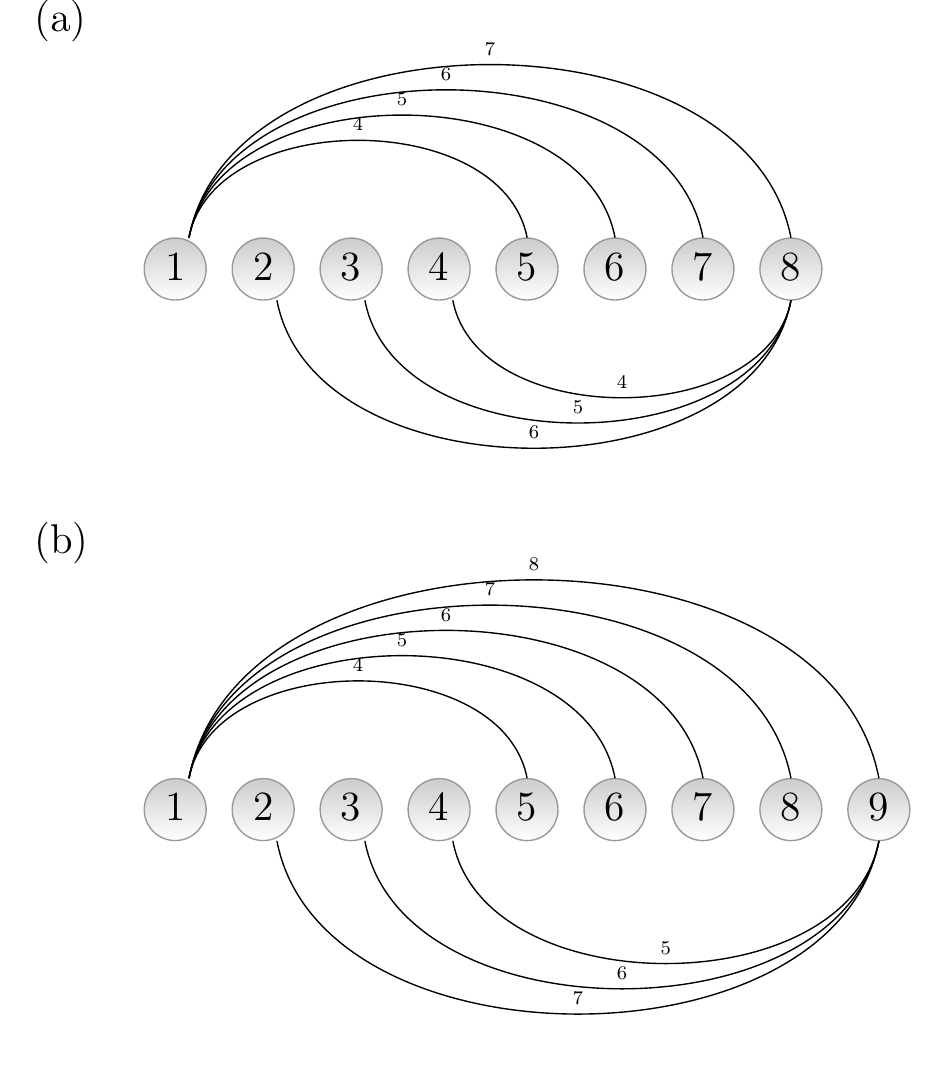} 
\caption{\label{bistar_appendix_figure} Linear arrangements of balanced bistar trees of $n$ vertices that maximize $D^t$. Vertex labels indicate the position of each vertex. Edge labels indicate edge distances. (a) $n = 8$ and $D^t = 37$. (b) $n = 9$ and $D^t = 48$. }
\end{figure}

Suppose that vertices are labelled with positions in the linear arrangement. An edge between vertices $i$ and $j$ is indicated by the unordered pair $\{i,j\}$. 
The problem of obtaining a tree that maximizes $D_{max}^t$ for any tree $t$ of $n$ vertices is equivalent to the problem of finding   
the maximum spanning tree of a complete graph where the weight of the edge $\{i,j\}$ is $|i - j|$, as each possible spanning tree bijectively corresponds to a linear arrangement of some tree of $n$ vertices, and the sum of weights corresponds to its value of $D$. 
We will show that a balanced bistar tree is the outcome of an algorithm that is based on Prim's algorithm to find the minimum spanning tree of a graph \cite{cormen:2001}. 
Prim's original algorithm solves a minimization problem. The maximization problem can be solved using the customary minimization version with edge weights defined as 
$n - |i - j|$. We use a variant of Prim's algorithm to solve the maximization problem that eases the proof:
\begin{enumerate}
\item
Initialize the tree $t$ with vertex $1$. 
\item
\label{growth_step}
Find the edge linking one vertex in $t$ and another vertex outside $t$ such that the weight is maximized. Add the edge (and the new vertex) to $t$.
\item
Repeat step \ref{growth_step} until $t$ has $n$ vertices.
\end{enumerate}

In the context of our application, i.e. the maximization of $D$ for any possible tree $t$ of $n$ vertices, this variant of Prim's algorithm becomes
\begin{enumerate}
\item
Initialize the tree $t$ with vertex $1$.
\item
Set $x$ to $2$ and $y$ to $n$.
\item
\label{growth_step_variant}
Compare the length of the edges $\{1, y\}$ and $\{x, n\}$.  
If the longest edge is $\{1, y\}$, add the edge (and vertex $y$) to $t$ and decrement $y$. Otherwise, add $\{x, n\}$ (and add $x$) to $t$ and increment $x$.
\item
Repeat step \ref{growth_step_variant} until $t$ has $n$ vertices.
\end{enumerate}
Notice that the vertices that do not belong to $t$ are in the interval $[x, y]$.
As for Step \ref{growth_step_variant}, notice that the longest edge liking one vertex in $t$, namely one vertex in $[1,n]\setminus[x,y]$, and another vertex outside $t$, namely one vertex in $[x,y]$, can only be $\{1, y\}$ or $\{x, n\}$.
It is easy to see that the execution of this algorithm produces edges that correspond to a balanced bistar tree (Table \ref{Prims_algorithm_table}) that is arranged linearly as in figure \ref{bistar_appendix_figure}.

\begin{table}
\centering
\begin{tabular}{llll}
Iteration & Edge        & Length  & $[x,y]$ \\
\hline
0         & -           & -       & $[2,n]$ \\
1         & $\{1,n\}$   & $n - 1$ & $[2,n-1]$ \\
2         & $\{1,n-1\}$ & $n - 2$ & $[2,n-2]$ \\
3         & $\{2,n\}$   & $n - 2$ & $[3,n-2]$ \\
4         & $\{1,n-2\}$ & $n - 3$ & $[3,n-3]$ \\
5         & $\{3,n\}$   & $n - 3$ & $[4,n-3]$ \\
...       & ...         & ...     & ... 
\end{tabular}
\caption{\label{Prims_algorithm_table} The edge added at every iteration, its length and $[x,y]$, the interval of vertex labels that do not belong to $t$ after adding the edge. }
\end{table}

\section{Validation}

\label{validation_appendix}

The main results of the article, namely Table \ref{summary_table} and the chains of inequalities in equations \ref{chain_of_inequalities_D_equation}, \ref{chain_of_inequalities_Delta_equation}, \ref{chain_of_inequalities_Gamma_equation}, \ref{chain_of_inequalities_zeta_score_min_equation} and \ref{chain_of_inequalities_zeta_score_max_equation} have been validated using a brute force procedure up to $n = 11$ inspired by that of \cite{Esteban2016a}. 
For a given $n$, the procedure calculates $D_{min}^t$ and $D_{max}^t$ for every distinct unlabelled tree and consists in generating all the $n^{n-2}$ labelled trees using Pr\"ufer codes as in \cite{Esteban2016a} while updating a two-level table containing the current value of $D_{min}^t$ and $D_{min}^t$ and a signature of the tree to speed up the tree isomorphism test \cite{Campbell1991a}. The signature of a tree is defined as a vector containing the canonical names \cite{Campbell1991a} of the trees rooted at each of the Jordan centers \cite{Hedetniemi1981a} of the original free tree. A tree has 1 or 2 Jordan centers \cite{Hedetniemi1981a}.   
For each labelled tree whose underlying unlabelled tree is $t$,
\begin{enumerate}
\item
$D^t$ is calculated interpreting vertex labels as vertex positions in the linear arrangement.
\item
The signature of $t$ for the test of tree isomorphism is calculated. 
\item
$t$ is searched in the collection of already visited unlabelled trees. The unlabelled trees are accessed using a two-level look-up table: first, by their value of $n\DegreeSecondMoment$ and second, by the degree spectrum. The frequency spectrum is a vector indicating the number of vertices of that have a certain degree $k$. Then, the corresponding unlabelled tree is found comparing all the stored trees with the same degree spectrum against the target tree using their respective signatures.       
\item 
If $t$ is new, then both $D_{min}^t$ and $D_{max}^t$ are set to $D^t$ temporarily. 
\item
If $t$ is not new, then $D_{min}^t$ and $D_{max}^t$ are updated based on $D^t$.
\end{enumerate}
At the end of the exploration, one has the exact value of $D_{min}^t$ and $D_{max}^t$ for every tree $t$.
As a sanity check, we verify that the number of labelled trees in the look-up table is the one expected by OEI A00055, \url{https://oeis.org/A000055}.
We also verify, for every tree $t$, that  
\begin{enumerate}
\item
$D_{min}^t$ coincides with the value obtained by the corrected version of Shiloach's algorithm \cite{Esteban2015a} as a sanity check.
\item
$D_{min}^t$ and $D_{max}^t$ match the predictions in Table \ref{summary_table} and satisfy the inequalities in \ref{chain_of_inequalities_D_equation},\ref{chain_of_inequalities_Delta_equation} and \ref{chain_of_inequalities_Gamma_equation}.
\end{enumerate}
Equations \ref{Delta_max_conjecture1_equation}, \ref{Delta_max_conjecture2_equation}, \ref{Gamma_max_conjecture_equation} and \ref{zeta_score_max_conjecture_equation} have been inferred using the procedure above.